\documentclass[a4paper,12pt]{article}
\pdfoutput=1 
\usepackage{geometry}
\usepackage[margin=10mm]{caption}

\usepackage[utf8]{inputenc}
\usepackage[english]{babel}
\usepackage{amssymb, amsmath}
\usepackage{amsthm}
\usepackage{amsfonts,bbold}
\usepackage{upgreek}
\usepackage{mathrsfs}
\usepackage{setspace}
\usepackage{graphicx}
\usepackage{subcaption}
\usepackage{slashed}
\usepackage[sorting=none,maxbibnames=99, backend=bibtex]{biblatex}
\bibliography{Bibliography}
\usepackage{array, booktabs}
\usepackage{lmodern}
\usepackage{tikz}
\usepackage{microtype}
\usepackage{hyperref}
\usepackage{authblk}
\usepackage{appendix}

\newtheorem{theorem}{Theorem}
\newtheorem{prop}{Proposition}
\newtheorem{lemma}{Lemma}[section]
\newtheorem{corollary}{Corollary}[lemma]
\newtheorem{corollarythm}{Corollary}[theorem]

\theoremstyle{definition}
\newtheorem{remark}{Remark}[section]

\makeatletter
\renewenvironment{proof}[1][\proofname]{%
   \par\pushQED{\qed}\normalfont%
   \topsep6\p@\@plus6\p@\relax
   \trivlist\item[\hskip\labelsep\bfseries#1\@addpunct{.}]%
   \ignorespaces
}{%
   \popQED\endtrivlist\@endpefalse
}

\newcommand{\CF}{{\mathcal F}}

\newcommand{\beq}{\begin{equation}}
\newcommand{\eeq}{\end{equation}}
\newcommand{\eeeem}{\end{multline}}
\newcommand{\bem}{\begin{multline}}
\newcommand{\bqa} {\begin{eqnarray}}
\newcommand{\eqa} {\end{eqnarray}}
\newcommand{\eps}{\varepsilon}

\newcommand{\bmul}{\begin{multline}}
\newcommand{\emul}{\end{multline}}
\DeclareMathOperator{\End}{End}

\DeclareMathOperator{\Hom}{Hom}

\DeclareMathOperator{\Id}{Id}

\DeclareMathOperator{\Ad}{Ad}
\DeclareMathOperator{\ad}{ad}

\def \ra {\rightarrow}

\newcommand{\CA}{{ \mathcal A}}
\newcommand{\CB}{{\mathcal B}}
\newcommand{\CalC}{{\mathcal C}}
\newcommand{\CG}{{\mathcal G}}

\newcommand{\CH}{{\mathcal H}}

\newcommand{\CU}{{\mathcal U}}
\newcommand{\CV}{{\mathcal V}}

\newcommand{\CW}{{\mathcal W}}
\newcommand{\CM}{{\mathcal M}}

\newcommand{\ZZ}{{\mathbb Z}}
\newcommand{\RR}{{\mathbb R}}
\newcommand{\CC}{{\mathbb C}}
\newcommand{\NN}{{\mathbb N}}

\newcommand{\SA}{{\mathscr A}}

\newcommand{\SAl}{{{\mathscr A}_\ell}}
\newcommand{\SAal}{{{\mathscr A}_{a\ell}}}
\newcommand{\SV}{{\mathscr V}}
\newcommand{\SW}{{\mathscr W}}

\newcommand{\OL}{{O(L^{-\infty})}}
\newcommand{\Or}{{O(r^{-\infty})}}

\newcommand{\On}{{O(n^{-\infty})}}

\newcommand{\Ups}{\Upsilon}

\newcommand{\tr}{\uptau}

\def \l {\left(}
\def \r {\right)}

\def \lal {\langle}
\def \ral {\rangle}

\title{A classification of invertible phases of bosonic quantum lattice systems in one dimension}

\author{Anton Kapustin, Nikita Sopenko, Bowen Yang \smallskip\\ 
{\it California Institute of Technology, Pasadena, CA 91125, USA}}

\begin{document}

\maketitle

\abstract{We study invertible states of 1d bosonic quantum lattice systems. We show that every invertible 1d state is in a trivial phase: after tensoring with some unentangled ancillas it can be disentangled by a fuzzy analog of a finite-depth quantum circuit. If an invertible state has symmetries, it may be impossible to disentangle it in a way that preserves the symmetries, even after adding unentagled ancillas. We show that in the case of a finite unitary symmetry $G$ the only obstruction is an index valued in degree-2 cohomology of $G$. We show that two invertible $G$-invariant states are in the same phase if and only if their indices coincide.}

\section{Introduction}

Gapped phases of lattice systems in 1d are well understood by now \cite{FidkowskiKitaev,chen2011classification}. Their classification depends on whether one considers bosonic or fermionic systems, as well as whether other symmetries are present. For example, without imposing any symmetry there is only one bosonic phase (the trivial one). Fermionic phases without symmetries beyond the fermion parity $\ZZ_2^F$  are classified by $\ZZ_2$, where the non-trivial element can be realized by the Kitaev chain. Bosonic phases with a unitary on-site symmetry $G$ are classified by the abelian group $H^2(G,U(1))$, while fermionic phases with a symmetry $G\times\ZZ_2^F$ are classified by $\ZZ_2\times {\mathcal F}$ where ${\mathcal F}$ is an extension of $H^1(G,\ZZ_2)$ by $H^2(G,U(1))$. An element of an abelian group describing the phase will be called its index. Most of these classification results have been obtained by assuming that every phase of a 1d lattice system can be described by an injective Matrix Product State (MPS). 

Although ground states of generic gapped Hamiltonians are not MPS, for  finite 1d systems one can find an MPS approximation of any such state. There is an efficient numerical procedure (DMRG) for doing this. Thus from a practical standpoint the situation in 1d is very satisfactory: given a gapped Hamiltonian, there is an algorithmic procedure for identifying its phase. But from a theoretical standpoint the MPS approach leaves much to be desired. For example, it is not obvious that the index determined through a highly non-unique MPS approximation depends only on the original state. Neither is it obvious that it is an invariant of the phase (i.e. that it depends only on an appropriately defined equivalence class of states).

Recently, Y. Ogata and collaborators \cite{ogata2019classification,Ogata2,bourne2020classification} (see also \cite{Moon}) developed an approach to the classification of phases of 1d systems which does not rely on using an injective MPS. Instead they work with arbitrary 1d states satisfying the split property. The split property is the statement that the weak closure of the algebra of observables on any half-line is a Type I von Neumann algebra.  In the bosonic case, this is equivalent to saying that the Hilbert space of the system can be written as a tensor product of Hilbert spaces $\CH_-\otimes\CH_+$, so that  observables localized on the positive (resp. negative) half-line act non-trivially only on $\CH_+$ (resp. $\CH_-$). The split property, while seemingly obvious, typically fails for gapless systems in infinite volume, because von Neumann algebras more exotic than algebras of bounded operators in a Hilbert space show up. Nevertheless, it was shown by T. Matsui \cite{Matsui} that states of 1d lattice systems satisfying the area law (and in particular all gapped ground states of local Hamiltonians \cite{hastings2007area}) satisfy the split property. It is possible in this approach to define an index of a state and show that it is invariant with respect to a natural notion of equivalence (automorphic equivalence). This property ensures that the index is unchanged under continuous deformations of the Hamiltonian which do not close the gap \cite{bachmann2012automorphic,moon2020automorphic}. 

In this paper we use a variant of this approach to classify bosonic phases of matter with a certain additional property ("invertibility", see the next paragraph). Our definition of a quantum phase of matter is along the lines of Refs. \cite{hastings2005quasiadiabatic,QImeetsQM} and is based on the notion that a quantum phase of matter is a pattern of entanglement in a many-body wave-function, a point of view expounded in the monograph Ref.~\cite{QImeetsQM}. In Ref. \cite{QImeetsQM} two states were said to be in the same phase if, after stacking with some unentangled state of an ancillary system, they could be connected by a finite-depth quantum circuit. For the case of systems with a unitary symmetry, one additionally requires the quantum circuit to preserve the symmetry. Our definition is the same, except we replace quantum circuits with their fuzzy analogs, locally-generated automorphisms (LGAs) \cite{ThoulessHall}. This notion of a phase does not make reference to Hamiltonians  and applies to arbitrary pure states of lattice systems. 

To focus on ground states of gapped Hamiltonians, Refs. \cite{ogata2019classification,Ogata2,bourne2020classification} impose the split property. In this paper we focus on states which are  ``invertible'' in the sense of A. Kitaev \cite{Kitaevlecture}. Loosely speaking, an invertible state is a state which lacks long-range order. More precisely, a system is in an invertible state if, when combined with some ancillary system, it can be completely disentangled by applying an LGA. The ancillary degrees of freedom may be entangled between themselves but not with the original system. We show that for states of 1d systems  invertibility implies the split property, but unlike the latter, it has a clear physical meaning and can be generalized to higher dimensions. 
Our definition of a quantum phase of matter applies both to bosonic and fermionic systems, but in this paper we focus on bosonic systems, since the fermionic case is more technically involved. 

Our main results are as follows. If no symmetry is imposed, then we show that every invertible bosonic 1d state is in the trivial phase. In the case when a finite unitary on-site symmetry $G$ is present, we define an index valued in an abelian group for arbitrary $G$-invariant invertible 1d systems and prove that it is a complete invariant of such phases. That is, two $G$-invariant invertible 1d states are in the same phase if and only if their indices coincide. Our  definition of the index is equivalent to that in \cite{ogata2019classification,Ogata2} but  is formulated in terms of properties of domain walls for symmetries. We also show that a non-trivial index is an obstruction for a system to have an edge which preserves both the symmetry $G$ and invertibility.

The organization of the paper is as follows. In Section 2 we formulate our definition of phases and invertible phases as suitable equivalence classes of pure states of lattice systems. In Section 3 we define an index for invertible 1d states with unitary symmetries and show that entangled pair states realize all possible values of the index. A reader who is mostly interested in the properties of systems without symmetries can skip this section and proceed to Section 4. In Sections 4.1 and 4.2 we study invertible states of 1d systems without symmetries. We show that any such state is in a trivial phase. In Section 4.3 we classify phases of invertible 1d states with a finite unitary symmetry. In Appendix A we show that $G$-invariant pure factorized states of 1d systems are all in the trivial phase according to our definition. In Appendix B we show that the index of a state can be determined provided one has access to a sufficiently large but finite piece of the system. This implies that a non-trivial index is an obstruction for a system to have an edge which preserves both the symmetry $G$ and invertibility.

\section{States and phases of 1d lattice systems}

A 1d system is defined by its algebra of observables $\SA$. In the bosonic case, this is a $C^*$-algebra defined as a completion of the $*$-algebra of the form
\beq\label{dirlim}
\SAl=\underset{M}\varinjlim\, \underset{j\in [-M,M]}\bigotimes \SA_{j},
\eeq
where $\SA_{j}=\End(\SV_j)$ for some finite-dimensional Hilbert space  $\SV_j$. The numbers $d_j={\rm dim} \SV_j$ are not assumed to be bounded, but we assume that $\log d_j$ grows at most polynomially with $|j|$. Elements of the algebra $\SA$ are called (quasi-local) observables. For any finite subset $\Gamma\subset\ZZ$ we define a sub-algebra $\SA_\Gamma=\otimes_{j\in\Gamma} \SA_{j}$. Its elements are called local observables localized on $\Gamma$. The union of all $\SA_\Gamma$ is the algebra of local observables $\SA_\ell$. 

There is a distinguished dense $*$-sub-algebra $\SAal$ of $\SA$ which we  call the algebra of almost local observables. $\CA\in\SA$ is an almost local observable if its commutators with observables in $\SA_{j}$ decay faster than any power of $|j|$. More precisely, there exists $k\in \ZZ$ and a monotonically-decreasing positive (MDP) function $h:\RR \ra \RR$ which decays faster than any power such that for any $\CB \in\SA_{j}$ one has $\|[\CA,\CB]\|\leq \|\CA\| \cdot \|\CB\| h(|j-k|).$ We will say that $\CA$ is $h$-localized on $k$.  

Almost local observables are the building blocks for Hamiltonians with good locality properties. Such a Hamiltonian is a formal linear combination
\beq
F=\sum_{j\in\ZZ} F_{j},
\eeq
where $F_{j}$ is a self-adjoint almost local observable which is $h$ localized on $j$ (with the same function $h$ for all $j$), and such that the sequence $\|F_{j}\|$, $j\in\ZZ$, is bounded. One may call a Hamiltonian of this kind an almost local Hamiltonian. Note that $F$ is not a well-defined element of $\SA$ unless the sum is convergent. To emphasize this, we prefer to call such a formal linear combination a 0-chain rather than a Hamiltonian.  Note also that $\ad_F(\CA)=[F,\CA]=\sum_{j} [F_{j},\CA]$ is a well-defined almost local observable for any $\CA\in\SAal$. It is easy to see that $\ad_F$ is an derivation of $\SAal$. 

The main use of 0-chains is to define a distinguished class of automorphisms of $\SA$ which preserve the sub-algebra $\SAal$. Consider a self-adjoint 0-chain $F(s)$ which depends continuously\footnote{We say that a 0-chain $F(s)$ depends continuously on $s$ if $F_{j}(s)$ depends continuously on $s$ for all $j\in\ZZ$.} on a parameter $s\in [0,1]$. Then we define a one-parameter family of automorphisms $\alpha_F(s)$, $s\in [0,1]$, as a solution of the differential equation
\beq
\frac{d}{ds} \alpha_F(s)(\CA)=\alpha_F(s)( i [F(s),\CA])
\eeq
with the initial condition $\alpha_F(0)=\Id$. One can show that this differential equation has a unique solution for any 0-chain $F(s)$ \cite{bratteli2012operator2}. It follows from Lieb-Robinson bounds that $\alpha_F(s)$ maps $\SAal$ to itself (see Lemma A.2 of \cite{ThoulessHall}). We will use a short-hand $\alpha_F(1)=\alpha_F$. 

Automorphisms of the form $\alpha_F$ for some $F(s)$ will be called locally-generated automorphisms (or LGAs). They are ``fuzzy" analogs of finite-depth unitary quantum circuits. Note that locally-generated automorphisms form a group. Indeed, the composition of $\alpha_F$ and $\alpha_G$ can be  generated by the 0-chain
\beq \label{eq:LGAcomposition}
G(s)+\alpha_G(s)^{-1}(F(s)).
\eeq
The inverse of $\alpha_F$ is an automorphism generated by 
\beq
-\alpha_F(s) (F(s)).
\eeq

A state on $\SA$ is a positive linear function $\psi:\SA\ra {\mathbb C}$ such that $\psi(\Id)=1$. Given a state $\psi$ and a locally-generated automorphism $\alpha_F$, we can define another state by $\alpha_F(\psi) := \psi \circ \alpha_F(\CA) = \psi(\alpha_F(\CA))$. We will say that $\alpha_F(\psi)$ is related by an LGA to $\psi$, or LGA-equivalent to $\psi$. We denote by $\Phi(\SA)$ the set of LGA-equivalence classes of pure states on $\SA$.

The notion of LGA-equivalence makes sense for arbitrary pure states of lattice systems. If the states in question are unique ground states of gapped almost local Hamiltonians, then they are LGA-equivalent if and only if the corresponding Hamiltonians can be continuously deformed into each other without closing the gap \cite{bachmann2012automorphic,moon2020automorphic}. It is also easy to see that if a pure state $\psi$ is a unique ground state of a gapped almost local Hamiltonian, then so is every state which is LGA-equivalent to $\psi$. Thus for gapped ground states LGA-equivalence  is the same as homotopy equivalence. More generally, 
it has been proposed to define a zero-temperature phase as an equivalence class of a (not necessarily gapped) pure state under the action of finite-depth quantum circuits \cite{QImeetsQM}. One obvious defect of such a definition is that quantum circuits cannot change the range of correlations by more than a finite amount. Thus according to this definition  ideal atomic insulators which have a finite correlation range are not in the same phase as any state with exponentially decaying correlations. Replacing finite-depth quantum circuits with LGAs fixes this defect. 

A state $\psi$ is called factorized if $\psi(\CA \CB)=\psi(\CA)\psi(\CB)$ whenever $\CA$ and $\CB$ are local observables supported on two different sites. It is easy to see that all factorized pure states of $\SA$ are LGA-equivalent.\footnote{It is important here that we are dealing with bosonic systems. The statement is not true in the fermionic case.} Thus factorized pure states on $\SA$ define a distinguished element $\tr(\SA)\in\Phi(\SA)$. Following \cite{QImeetsQM}, we will say that a state on $\SA$ is Short-Range Entangled (SRE) if it is LGA-equivalent to a pure factorized state on $\SA$.


To compare states on different 1d lattice systems, we define the notion of stable equivalence following A. Kitaev \cite{Kitaevlecture}. We say that a pure state $\psi_1$ on a system $\SA_1$ is stably equivalent to a pure state $\psi_2$ on a system $\SA_2$ if there exist systems $\SA_1'$ and $\SA_2'$ and factorized pure states $\psi_1'$ and $\psi_2'$ on them such that the pairs $(\SA_1\otimes\SA_1',\psi_1\otimes\psi_1')$ and
$(\SA_2\otimes\SA_2',\psi_2\otimes\psi_2')$ are LGA-equivalent. 

A 1d phase is defined to be a stable equivalence class of a pure state $\psi$ on a 1d system $\SA$.  We will denote the set of (1d) phases by $\Phi$. Any two factorized pure states on any two 1d systems are stably equivalent. Thus it makes sense to say that the stable equivalence class of factorized pure states is the trivial phase $\tr\in\Phi$. It is also easy to see that the tensor product of systems and their states descends to a commutative associative binary operation $\otimes$ on $\Phi$, and that $\tr$ is the  neutral element with respect to this operation. In other words, $(\Phi,\otimes,\tr)$ is a commutative monoid.  

Next we define what we mean by an invertible phase. Following A. Kitaev \cite{Kitaevlecture}, we will say that a pure state $\psi$ on $\SA$ is in an invertible phase if there is a 1d system $\SA'$ and a pure state $\psi'$ on it such that $\psi\otimes\psi'$ is in a trivial phase. In other words, $x\in\Phi$ is an invertible phase if it has an inverse. Therefore invertible phases form an abelian group which we denote $\Phi^*.$

Note that we defined a phase without any reference to a Hamiltonian. In general, given a pure state, it is not clear whether it is a ground state of any almost-local Hamiltonian. 

It will be shown below that all invertible 1d systems are in a trivial phase\footnote{This is not true for fermionic systems, the Kitaev chain is a counter-example.}, i.e. $\Phi^*$ consists of a single element $\tr$. To get a richer problem, we will study systems and phases with on-site symmetries. Recall that we assumed that each $\SA_j$ is a matrix algebra, i.e. it has the form $\SA_j=\End(\SV_j)$ for some finite-dimensional vector space $\SV_j$ (the ``on-site Hilbert space"). One says that a group $G$ acts on a system $\SA$ by unitary on-site symmetries if one is given a sequence of homomorphisms $R_j:G\ra U(\SV_j)$,
$j\in\ZZ$, where $U(\SV_j)$ is the unitary group of $\SV_j$. Since each $\SV_j$ is finite-dimensional, it decomposes as a sum of a finite number of irreducible representations of $G$. 

Homomorphisms $R_j$ give rise to a $G$-action on $\SA$, $\SAl$ and $\SAal$: $(g, \CA)\mapsto \Ad_{R(g)}(\CA)=R(g)\CA R(g)^{-1}$, where
\beq
R(g)=\prod_{j\in\ZZ} R_{j}(g),\quad g\in G.
\eeq
The latter is a formal product of unitary local observables $R_{j}(g)\in \SA_j$ such that the  automorphism $\Ad_{R(g)}$ is well-defined. A state $\psi$ on a system $\SA$ with an  on-site action of $G$ is said to be $G$-invariant if it is invariant under $\Ad_{R(g)}$ for all $g\in G$: $\psi \circ \Ad_{R(g)}=\psi$ for all $g\in G$. 

When defining LGA-equivalence of $G$-invariant states on $\SA$, one needs to restrict to those LGAs which are generated by $G$-invariant self-adjoint  0-chains, i.e. self-adjoint 0-chains $F(s)=\sum_{j} F_{j}(s)$ such that all observables $F_{j}(s)$ are $G$-invariant. As a result, the automorphism it generates is called $G$-equivariant as it commutes with the action of the group. We denote by $\Phi_G(\SA)$ the set of $G$-equivariant LGA-equivalence classes of $G$-invariant pure states of $\SA$ (with some particular $G$-action which is not indicated explicitly). 

There is a distinguished element $\tr_G\in \Phi_G(\SA)$, but its definition is not completely obvious. The most general $G$-invariant factorized pure state on $\SA$ is uniquely defined by the condition that when restricted to $\SA_j$ it takes the form 
\beq\label{factorizedpure}
\psi(\CA_j)=\langle v_j|\CA_j|v_j\rangle, \quad \CA_j\in \SA_j,
\eeq
where $v_j\in\SV_j$ is a unit vector which transforms in a one-dimensional representation of $G$. Not every $\SA$ admits such states. If such states on $\SA$ exist, they need not all correspond to the same element in $\Phi_G(\SA)$.  For example, consider a system where $\SV_j$ is the trivial one-dimensional representation of $G=\ZZ_2$ for all $j$ except $j=0$, while $\SV_0$ is a sum of the trivial and the non-trivial one-dimensional representations with normalized basis vectors $e_+$ and $e_-$. There are two different $\ZZ_2$-invariant factorized pure states corresponding to $v_0=e_+$ and $v_0=e_-$, and they are clearly not related by a $G$-equivariant LGA. 
We will call a state of the form (\ref{factorizedpure}) where all $v_j$ are $G$-invariant a special $G$-invariant factorized pure state. 
We will define $\tr_G\in\Phi_G(\SA)$ to be the $G$-equivariant LGA-equivalence class of a special $G$-invariant factorized pure state. This definition makes sense because for any particular $\SA$ all special $G$-invariant factorized pure states are in the same $G$-equivariant LGA-equivalence class.

Next we define the notion of a $G$-invariant phase.
We say that a $G$-invariant pure state $\psi_1$ of a system $\SA_1$ is $G$-stably equivalent to a pure state $\psi_2$ on a system $\SA_2$ if there exist systems $\SA_1'$ and $\SA_2'$ and special $G$-invariant factorized pure states $\psi_1'$ and $\psi_2'$ on them such that the pairs $(\SA_1\otimes\SA_1',\psi_1\otimes\psi_1')$ and
$(\SA_2\otimes\SA_2',\psi_2\otimes\psi_2')$ are related by a $G$-equivariant LGA. 
A $G$-invariant phase is defined to be a $G$-stable equivalence class of a $G$-invariant pure state $\psi$ of a 1d system $\SA$.  We will denote the set of $G$-invariant phases by $\Phi_G$. $\Phi_G$ is a commutative monoid, with the trivial $G$-invariant  phase as the neutral element $\tr_G$. Invertible elements in this monoid form an abelian group which we denote $\Phi^*_G$. Elements of this group will be called $G$-invertible phases. 

It is shown in Appendix A that any $G$-invariant factorized pure state (that is, a state of the form (\ref{factorizedpure}) where all $v_j \in \SV_j$ transform in one-dimensional representations of $G$) belongs to the trivial $G$-invariant phase. This provides a sanity check on our definition of a phase: from a physical viewpoint, a nontrivial state must exhibit some entanglement and cannot be factorized.

In this paper we will be studying $G$-invariant phases which are also invertible, that is, $G$-invariant phases which are mapped to $\Phi^*$ by the forgetful homomorphism  $\Phi_G\ra\Phi$. It is easy to see that every $G$-invertible phase is invertible. Less obviously, we will show that every $G$-invariant phase which maps to $\Phi^*$ is $G$-invertible. We will define an index taking values in an abelian group which classifies such phases. This abelian group is nothing but  $\Phi^*_G\subset \Phi_G$.

\section{An index for invertible 1d systems}

\subsection{Domain wall states}

Consider a $G$-invariant pure state $\psi$ of $\SA$. For any $g\in G$ we let
\beq
\rho_{> j}^g = \prod_{k > j} \Ad_{R_{k}(g)}.
\eeq
which is a restriction of a locally generated automorphism $\Ad_{R(g)}$ to a half-line $j > 0$. Clearly, $\rho_{> j}^g \circ \rho_{> j}^h=\rho_{> j}^{gh}$. We define the domain wall state $\psi^g_{>j}$ by
\beq
\psi^g_{>j}(\CA )=\psi(\rho_{> j}^g(\CA )).
\eeq

From now on we suppose that $\psi$ defines a $G$-invariant invertible phase. Thus there exists a system $\SA'$, a factorized pure state $\psi'$ on $\SA'$, and a locally-generated automorphism $\alpha$ of $\SA \otimes \SA'$ that transforms the state $\Psi = \psi \otimes \psi'$ to a factorized pure state $\Omega$ on $\SA \otimes \SA'$, i.e. $\Psi = \Omega \circ \alpha$. We can extend the action of $\rho^g$ on $\SA$ to an action on $\SA \otimes \SA'$ by making it trivial on $\SA'$. Then the automorphism
\beq
\tilde{\rho}^g = \prod_{k \in \ZZ} \Ad_{\alpha(R_{k}(g))}.
\eeq
preserves the state $\Omega$, and its restriction to a half-line $k>j$ is given by $\tilde{\rho}^g_{>j} = \alpha \circ \rho^g_{>j} \circ \alpha^{-1}$.

We will need the following result from Appendix C of \cite{ThoulessHall} whose proof we reproduce here for convenience.
\begin{prop} \label{prop:split}
Let $\Omega$ be a factorized pure state on $\SA$. Let $\alpha$ be an automorphism of $\SA$ which asymptotically preserves $\Omega$, in the sense that there exists a monotonically decreasing positive function $h(r)=\Or$ on $[0,\infty)$ such that for any $\CA $ localized on $[r,\infty)$ or $(-\infty,-r]$ one has $|\Omega(\alpha(\CA ))-\Omega(\CA )|\leq h(r) \|\CA \|$. Then the state $\Omega^\alpha:\CA \mapsto\Omega(\alpha(\CA ))$ is unitarily equivalent to $\Omega$.
\end{prop}
\begin{proof}
Let $\SA_+$ and $\SA_-$ be the $C^*$-subalgebras of $\SA$ corresponding to $j> 0$ and $j\leq 0$, respectively. Let $\Omega_\pm$ be the restriction of $\Omega$ to $\SA_\pm$, and $\Omega^\alpha_\pm$ be the restriction of $\Omega^\alpha$ to $\SA_\pm$. By Cor. 2.6.11 in \cite{bratteli2012operator}, the state $\Omega^\alpha_+$ is quasi-equivalent to $\Omega_+$ and the state $\Omega^\alpha_-$ is quasi-equivalent to $\Omega_-$. Therefore $\Omega=\Omega_+\otimes\Omega_-$ is quasi-equivalent to $\Omega^\alpha_+\otimes\Omega^\alpha_-$. Since both $\Omega$ and $\Omega^\alpha$ are pure states, it remains to show that $\Omega^\alpha$ is quasi-equivalent to $\Omega^\alpha_+\otimes\Omega^\alpha_-$.

Let $\rho_{j},\rho^{\alpha}_{j}\in\SA_{j}$ be the density matrices for the restriction of $\Omega$ and  $\Omega^\alpha$ to $\SA_{j}$. The density matrix $\rho_{j}$ is pure, while $\rho^{\alpha}_{j}$ is mixed, in general. But since for any $\CA \in\SA_{j}$ we have $|{\rm Tr}(\rho^{\alpha}_{j}-\rho_{j})\CA |\leq h(|j|)\|\CA \|$, we have $\|\rho_{j}-\rho^{\alpha}_{j}\|_1\leq h(|j|),$ where $\|\cdot\|_1$ is the trace norm. Thus for large $|j|$ the entropy $S^{\alpha}_{j}$ of $\rho^{\alpha}_{j}$ rapidly approaches zero. Specifically, by Fannes' inequality, for all sufficiently large $|j|$ we have $S^{\alpha}_{j}\leq h(|j|) \log(d_{j}/h(|j|))$. Here $d_{j}^2=\dim\SA_{j}$. Since we assumed that $\log d_{j}$ grows at most polynomially with $|j|$, $S^\alpha=\sum_{j} S_{j}^\alpha <\infty$. Therefore the entropy of the restriction of $\Omega^\alpha$ to any finite region of $\ZZ$ is upper-bounded by $S^\alpha$. By Proposition 2.2 of \cite{matsui2001split} (where the proof of equivalence of (i) and (ii) does not use translation invariance) and Theorem 1.5 of \cite{Matsui}, the state $\Omega^\alpha$ is quasi-equivalent to $\Omega^\alpha_+\otimes\Omega^\alpha_-$.
\end{proof}
The automorphism $\rho^g_{>j}$ does not preserve $\psi$, but it asymptotically preserves it, by Lemma A.4 of \cite{ThoulessHall}. Therefore $\tilde\rho^g_{>g}$ asymptotically preserves $\Omega$. 
Then the above proposition implies that the state $\Omega^g_{>j}$ defined by $\Omega^g_{>j}(\CA )=\Omega(\tilde{\rho}^g_{>j}(\CA ))$ is unitarily equivalent to $\Omega$. Since $\Omega \circ \alpha=\Psi$ and $\rho^g_{>j}$ acts non-trivially only on $\SA$, the states $\psi$ and $\psi^g_{>j}$ are also unitarily equivalent.

Let $(\Pi, \CH, |0\rangle)$ be the GNS data of the original state $\psi$ constructed in the usual way. We can identify $\Pi(\rho^g_{>j}(\CA ))$ with the GNS representation of $\psi^g_{>j}.$ Specifically,
let $(\Pi_{>j}^g, \CH_{>j}^g, |0\rangle_{>j}^g)$ be the GNS data of $\psi_{>j}^g$. In particular, this coincides with $(\Pi, \CH, |0\rangle)$ when $g$ is the identity. $\CH_{>j}^g$ is the completion of $\SA/I^g_{>j}$ where the left ideal
 \beq
 I^g_{>j}:=\{\CA | \psi(\rho^g_{>j}(\CA ^*\CA ))=0\}.
 \eeq
 Since $\rho^g_{>j}(I_{>j}^h)=I^{hg^{-1}}_{>j}$, $\rho^g_{>j}$ induces a linear isometry of Hilbert spaces $$\iota_{>j}^g:\CH^h_{>j}\rightarrow\CH^{hg^{-1}}_{>j},$$  such that $\iota^g_{>j}\circ \iota^h_{>j}=\iota^{gh}_{>j}$ for all $g,h\in G.$ Note $\iota$ does not have a fixed source or target. It is a (categorified) group action over all Hilbert spaces $\CH_{>j}^g$. It then follows that 
 \beq
 \iota^g_{>j}\Pi^g_{>j}(\CA )(\iota^g_{>j})^{-1}=\Pi(\rho^g_{>j}(\CA )). \label{CatAction}
 \eeq
 This, together with the unitary equivalence proven above, implies that for any $g\in G$ and any $j\in\ZZ$ there exists a unitary operator $U^g_{>j}$ such that
 
 \beq \label{Unitary}
\Pi(\rho^g_{>j}(\CA )) = U^g_{>j} \Pi(\CA ) \left(U^g_{>j}\right)^{-1} . 
\eeq
Since the commutant of $\Pi(\SA)$ consists of scalars, this equation defines $U^g_{>j}$ up to a multiple of a complex number with absolute value $1$.
 
Finally, we are ready to define the index. From equation (\ref{Unitary}) and $\rho_{> j}^g \circ \rho_{> j}^h=\rho_{> j}^{gh}$, we infer that
\beq
U_{>j}^gU_{>j}^h\Pi(\CA )(U_{>j}^gU_{>j}^h)^{-1}=U_{>j}^{gh}\Pi(\CA )(U_{>j}^{gh})^{-1}.
\eeq
Since the commutant of $\Pi(\SA)$ consists of scalars, it follows that
\beq\label{2cocycle}
U_{>j}^g U^h_{>j}= \nu_{>j}(g,h) U^{gh}_{>j},
\eeq
where $\nu_{>j}(g,h)$ is a complex number with absolute value $1$. Associativity of operator product implies that $\nu_{>j}(g,h)$ is a $2$-cocycle of the group $G$. Its cohomology class is independent of the different choices of $\{U_{>j}^g: g\in G\}$. We claim this cohomology class is an  index for bosonic systems. Firstly we show that it is independent of $j$. Then we demonstrate that any $G$-equivariant locally generated automorphism preserves the index. Lastly, we verify that the index of a state in the trivial phase is trivial. 

Given $i<j$, we let
\beq
R_{(i,j]}^g=\prod_{i< k\leq j}R_k(g).
\eeq
By (\ref{Unitary}), $\Pi(R_{(i,j]}^h)=\Pi(\rho_{>j}^g(R_{(i,j]}^h))=U_{>j}^g\Pi(R_{(i,j]}^h)(U_{>j}^g)^{-1}$, i.e. $\Pi(R_{(i,j]}^h)$ and $U_{>j}^g$ commute. Also it is easy to see that 
\beq
U_{>i}^g=\mu(g)\Pi(R_{(i,j]}^g)U_{>j}^g,
\eeq
where $\mu(g)$ is a complex number with absolute value  $1$. Therefore,
\begin{align}
    &\nu_{>i}(g,h)\mu(gh)\Pi(R_{(i,j]}^{gh})U^{gh}_{>j}\\\notag
    =&\nu_{>i}(g,h)U_{>i}^{gh}=U_{>i}^gU_{>i}^h\\\notag
    =&\mu(g)\Pi(R_{(i,j]}^g)U^{g}_{>j}\mu(h)\Pi(R_{(i,j]}^h)U^{h}_{>j}\\\notag
    =&\mu(g)\mu(h)\Pi(R_{(i,j]}^{gh})U^{g}_{>j}U^h_{>j}\\\notag
    =&\mu(g)\mu(h)\nu_{>j}(g,h)\Pi(R_{(i,j]}^{gh})U^{gh}_{>j}.
\end{align}
It is important here that $R_j(g)$ is an ordinary (non-projective) representation of $G$. From now on, we fix a domain wall position $j$ as it does not affect the index. We also omit this choice from all notations, for example $\rho_{>j}^g$ becomes simply $\rho^g_{>}.$ 

Next, we prove that the index is invariant under an automorphism $\beta:=\beta(1)$ generated by a $G$-equivariant self-adjoint 0-chain $F(t)=\sum_{j} F_{j}(t)$. Note the site $j$ here is not the domain wall position which has been fixed. In general $\rho_>^g\circ \beta\ne\beta\circ \rho_>^g$, however their commutator $\beta\circ \rho_>^g\circ\beta^{-1}\circ\rho_>^{g^{-1}}=\Ad_{\CB _g}$ for an almost local observable $\CB _g$. Furthermore,
\begin{equation} \label{identity}
    \CB _{gh}=\CB _g\rho_>^g(\CB _h).
\end{equation}
Before giving the proof, we notice that these facts indeed lead us to the desired invariance of cohomology.

Define a new state 
\beq
\psi^\beta(\CA ):=\psi(\beta(\CA )).
\eeq
Repeating the argument above for state $\psi^\beta$, we get \begin{equation}
    \Pi(\beta(\rho_>^g(\CA )))=W_>^g\Pi(\beta(\CA ))(W_>^g)^{-1}.
\end{equation}
However,
\begin{align}
    &\Pi(\beta(\rho_>^g(\CA )))\\\notag
    =&\Pi(\beta\circ \rho_>^g\circ\beta^{-1}\circ\rho_>^{g^{-1}}(\rho_>^g(\beta(\CA ))))\\\notag
    =&\Pi(\CB _g)U_>^g\Pi(\beta(\CA ))(U_>^g)^{-1}\Pi(\CB _g)^{-1}.
\end{align}
Therefore, \begin{equation}
    W_>^g=\mu(g)\Pi(\CB _g)U_>^g,
\end{equation}
where $\mu(g)$ is a complex number with norm $1$. And finally,
\begin{align}
    W_>^gW_>^h=&\mu(g)\mu(h)\Pi(\CB _g)U_>^g\Pi(\CB _h)U_>^h\\\notag
    =&\mu(g)\mu(h)\Pi(\CB _g\rho_>^g(\CB _h))U_>^gU_>^h\\\notag
    =&\mu(g)\mu(h)\mu(gh)^{-1}\nu(g,h)W^{gh}.
\end{align}

To prove that such $\CB _g$ exists, notice that the automorphism $\beta(t)\circ \rho_>^g\circ\beta^{-1}(t)\circ\rho_>^{g^{-1}}$ is generated by the $0$-chain $F_g(t):=\rho_>^g\circ \beta(t)(\rho_>^{g^{-1}}(F(t))-F(t))$, which is almost local.
For any almost local self-adjoint $\CA (t)$ depending continuously on $t$, we can define an almost local unitary  observable $E_\CA (t)$  satisfying the equation
\beq
-i\frac{d}{dt}E_\CA (t)=E_\CA (t)\CA (t). \label{ordExp}
\eeq
When viewed as a $0$-chain, the automorphism generated by $\CA (t)$ is simply $\Ad_{E_\CA (t)}$. We let $E_g(t):=E_{F_g}(t)$, then it follows that $\CB _g=E_{g}(1).$ 
To establish identity (\ref{identity}), we need two properties of the observables $E_\CA (t)$:
\begin{equation}
    \alpha(E_\CA (t))=E_{\alpha(\CA )}(t),
\end{equation}
and
\begin{equation}
    E_\CA (t)E_\CB(t)=E_{\mathcal C}(t),
\end{equation}
where $\mathcal C(t)=\alpha_\CB^{-1}(t)(\CA (t))+\CB(t).$ Both are easily checked from equation (\ref{ordExp}). 

Then
\begin{align}
    E_{g}(t)\cdot\rho_>^g(E_h(t))
    =\rho_>^g(E_G(t)\cdot E_h(t))
    =\rho_>^g(E_X(t)).
\end{align}
By the first property, $G(t)=\beta(t)(\rho_>^{g^{-1}}(F(t))-F(t)).$ By the second property,
\begin{align}
    X(t)=&\alpha_{F_h}^{-1}(t)(G(t))+F_h(t)\\
    =&\rho_>^h\circ \beta(t)\circ \rho_>^{h^{-1}} \circ \beta^{-1}(t) \circ \beta(t)\big(\rho_>^{g^{-1}}(F(t))-F(t)\big)\\
    &+\rho_>^h\circ \beta(t)(\rho_>^{h^{-1}}(F(t))-F(t))\\
    =&\rho_>^h\circ \beta(t)\big(\rho_>^{(gh)^{-1}}(F(t))-F(t)\big).
\end{align}
Therefore, $\rho_>^g(X(t))=F_{gh}(t)$, which is what we need.

At last, we verify that the index of the trivial phase is trivial. By the previous result, it suffices to compute the index of a factorized state. Since $\psi(\rho_>^g(\CA ))=\psi(\CA )$ for any factorized $G$-invariant pure state $\psi$, the GNS representation $\Pi=\Pi^g$. Equations (\ref{CatAction}) and (\ref{Unitary}) imply
\beq
\iota^g_{>}\Pi(\CA )(\iota^g_{>})^{-1}=\iota^g_{>}\Pi^g_{>}(\CA )(\iota^g_{>})^{-1}=\Pi(\rho_{>}^g(\CA ))=U_>^g\Pi(\CA )(U_>^g)^{-1}.
\eeq
As $\iota^g_{>} \iota^h_{>}=\iota^{gh}_{>}$, the index is trivial.

Taken together, these results imply that the index of an invertible $G$-invariant pure state on $\SA$ depends only on its $G$-invariant  phase.

\begin{remark}\label{rem:Ug}
Let $\SA_{>j}=\otimes_{k>j}\SA_k$ and $\SA_{\leq j}=\otimes_{k\leq j}\SA_k$. We can define commuting von Neumann algebras acting in $\CH_\psi$ by letting $\CM_{>j}=\Pi_\psi(\SA_{>j})''$ and $\CM_{\leq j}=\Pi_\psi(\SA_{\leq j})''.$  These two algebras are each other commutants and generate the whole $B(\CH_\psi)$. Then the definition of $U^g_{>j}$ implies that $U^g_{>j}\in \CM_{>j}$. Similarly, replacing in (\ref{Unitary}) the automorphism  $\rho^g_{>j}$ with $\rho^g_{\leq j}=\rho^g \left(\rho^g_{>j}\right)^{-1}$, we can define a unitary  $U^g_{\leq j}\in\CM_{\leq j}.$ Then it is easy to see that the index we defined above is the same as the index defined in \cite{ogata2019classification}.\end{remark}
 
The stacking law agrees with the group structure of $H^2(G, U(1))$. This is seen by taking the domain wall unitary of the stacked system to be the tensor product of the domain wall unitary operators of the respective subsystems. According to (\ref{2cocycle}), the stacked 2-cocycle is the product of the 2-cocycles of the subsystems. 

An important property of the index is that it is locally computable: it can be evaluated approximately given the restriction of the state to any sufficiently large segment of the lattice. This is shown in Appendix B. Local computability of the index implies there can be no $G$-invariant invertible interpolation between a state with a non-trivial index and a $G$-invariant factorized pure state. Put more concisely, an invertible $G$-invariant system with a non-trivial index cannot have a $G$-invariant non-degenerate edge.

\subsection{Examples}

In this section we compute the index of some standard examples of non-trivial invertible states. 

In the bosonic case, Matrix Product States furnish examples of invertible states invariant under a (finite) symmetry group $G$. To construct such an example, we pick a projective unitary representation $Q$ of $G$ on a finite-dimensional Hilbert space $\SW$. Thus we are given unitary operators $Q(g)\in U(\SW)$, $g\in G$, satisfying
\beq
Q(g) Q(h)=\nu(g,h) Q(gh),
\eeq
where $\nu(g,h)$ is a 2-cocycle with values in $U(1)$. We take the local Hilbert space $\SV_j$ to be $\SW_j\otimes\SW_j^*$, where $\SW_j$ is isomorphic to $\SW$ for all $j\in\ZZ$. Then $G$ acts on $\SV_j$ via $R(g)=Q(g)\otimes Q(g)^*$. This is an ordinary (non-projective) action. 

To define a $G$-invariant state on $\SA$, we first specify it on $\SA_l$ and then extend by continuity. We note first that
\beq
\SAl=\otimes_{k\in\ZZ} \End(\SW_k\otimes\SW_k^*)=\otimes_{k\in\ZZ} \End(\SW^*_k\otimes\SW_{k+1}).
\eeq
It is understood here that in both infinite tensor product all but a finite number of elements are identity elements.
Thus we can get a $G$-invariant pure state on $\SAl$ by picking a $G$-invariant vector state on $\End(\SW_k^*\otimes\SW_{k+1})$ for all $k\in\ZZ$. If the representation $Q$ is irreducible, there is a unique choice of such a state: the one corresponding to the vector $\frac{1}{\sqrt d}\, 1_\SW\in \SW^*\otimes\SW$, where $d$ is the dimension of $\SW$. In the physics literature such a state on $\SAl$ is known as an entangled-pair state. Let us denote this state $\psi$. 

Note that $\psi$ is in an invertible phase. Indeed, any vector state on $\End(\SW_k^*\otimes\SW_k)$ can be mapped to a factorized vector state by a unitary transformation. The product of these unitary transformations for all $k$ gives us a locally-generated automorphism connecting $\psi$ with a factorized pure state on $\SA$. In fact, it is easy to see that the $G$-invariant phase of $\psi$ is $G$-invertible. The inverse system is obtained by replacing $\SW$ with $\SW^*$ and $Q$ with $Q^*$.

To describe the GNS Hilbert space corresponding to $\psi$, let us pick an orthonormal basis $|n\rangle$, $n=1,\ldots,d,$ in $\SW$ and denote by $|j,n\rangle$, $\langle j,n|$ the corresponding orthonormal basis vectors for $\SW_j$ and $\SW^*_j$. Let us also denote by $|1'_j\rangle\in \SW_j^*\otimes \SW_{j+1}$ the vector $\frac{1}{\sqrt d}\sum_{n=1}^d \langle j,n|\otimes |j+1,n\rangle $. Then the GNS Hilbert space is the completion of the span of vectors of the form 
\beq
\left(\bigotimes_{k<J} |1'_k\rangle \right)\otimes \langle J,n_J| \otimes | J+1,m_J\rangle \otimes \ldots \otimes \langle J', n_{J'}|\otimes |J'+1,m_{J'}\rangle \otimes \left(\bigotimes_{l>J'} |1'_l\rangle \right),
\eeq
where $J,J'$ ($J\leq J'$) are integers. 

It is easy to check that the operator $U^g_{>j}$ is given (up to a scalar multiple) by 
\beq
U^g_{>j}=\bigotimes_{k\leq j} 1_k\otimes Q_{j+1}(g)\otimes \bigotimes_{l>j} R'_l(g),
\eeq
where $Q_{j+1}(g)\in \End(\SW_{j+1})$ is given by $a\mapsto Q(g) a$, $R'_l(g)\in End\left(\SW_l^*\otimes\SW_{l+1}\right)$ is given by $a\otimes b\mapsto Q(g)^*a\otimes Q(g)b$, and $1_k\in \End(\SW_k\otimes\SW^*_{k+1})$ is the identity operator. This is a well-defined operator because $|1'_l\rangle$ is invariant under $R'_l(g)$. Now, since the operators $R'_l(g)$ define an ordinary (non-projective) representation of $G$, we see that the index of the state $\psi$ is given precisely by the 2-cocycle $\nu(g,h)$.

\section{A classification of invertible phases of bosonic 1d systems}

\subsection{Preliminaries}

In the following, for a state $\psi$ we denote by $\CH_{\psi}$ and $\Pi_{\psi}$ the corresponding GNS Hilbert space and GNS representation. For a region $A$ of the lattice, we denote by $\psi|_A$ the restriction of a state $\psi$ to $\SA_{A}$.

An LGA $\alpha$ generated by $F$ which is $f$-local can be represented by an ordered conjugation with $\overrightarrow{\prod}_{j \in \Lambda} e^{i G_j}$ for almost local observables $G_j$ which are $g$-localized for some $g(r)$ that depends on $f(r)$ only. Indeed, suppose we have an $f$-localized 0-chain. The automorphism
\beq
\alpha_{F_{(-\infty, k+1]}} \circ \l \alpha_{F_{(-\infty, k]}} \r^{-1}
\eeq
is a conjugation by an almost local unitary $e^{i G_j}$ which is $g$-localized for some $g(r)$ that depends on $f(r)$ only. 



We define a restriction $G=F|_A$ of a 0-chain $F$ to a region $A$ by 
\beq
G_j = \int \prod_{k\in \bar{A}} d V_k \Ad_{\prod_{k \in \bar{A}} V_k} (F_j),
\eeq
where the integration is over all on-site unitaries $V_k \in \SA_k$ with Haar measure. Note that $[G,\CA]=0$ for any $\CA \in \SA_{\bar{A}}$, i.e. $G=F|_A$ is localized on $A$. It is also easy to see that $F-F|_A$ is almost localized on $\bar A$.


These properties of restriction have several immediate  consequences. Let $A$ and $B$ be a left and a right half-chain, correspondingly. We can represent $F$ as 
\beq
F = F|_{A} + \CF_0 + F|_B 
\eeq
for some almost local observable $\CF_0$. By Lemma A.4 from \cite{ThoulessHall} we can represent $\alpha_F$ as a product 
\beq \label{eq:autdecomposition}
\alpha_{F} = \alpha_{F|_A} \circ \alpha_{F|_B} \circ \Ad_{\CU_0}
\eeq
for some almost local unitary observable $\CU_0$. Similarly, let $A$, $B$ and $C$ be three regions $(-\infty,j)$, $[j,k]$ and $(k,\infty)$, correspondingly, for some sites $j$ and $k$. Then Lemma A.4 from \cite{ThoulessHall} implies 
\beq \label{eq:autdecomposition2}
\alpha_F = \alpha_{F|_A} \circ \alpha_{F|_C} \circ \alpha_{F|_B} \circ \Ad_{\CU_j} \circ \Ad_{\CU_k}
\eeq
for some almost local unitaries $\CU_j$ and $\CU_k$ $g$-localized at sites $j$ and $k$, correspondingly. Here $g(r)$ does not depend on $j$ and $k$.

Here and below we will denote by $\Gamma_r(j)$ the interval $[j-r,j+r]$. If $j=0$, we use a shorthand $\Gamma_r$.
Let $f:\ZZ_+\ra\RR_+$ be an MDP function such that  $f(r)=\Or$. Let $A$ be a subset of $\ZZ$ and $j\in \ZZ$. We will say that two states $\psi_1$ and $\psi_2$ on $\SA$ are $f$-close on $A$ far from $j$ if for any observable $\CA$ localized on $A\cap \bar\Gamma_r(j)$ we have $|\psi_1(\CA) - \psi_2(\CA) | \leq \|\CA\| f(r)$. Note that since all our algebras of observables are simple, all representations of $\SA$ are faithful, and thus $\|\CA\|$ can be replaced with $\|\Pi(\CA)\|$ in any representation $\Pi$ of $\SA$.

\begin{lemma}\label{lma:almcstates}
Let $\psi$ be a pure state on $\SA$ which is $f$-close on $\ZZ$ far from $j=0$ to a pure factorized state $\psi_0$ for some MDP function $f(r)=\Or$. Then $\psi$ and $\psi_0$ are unitarily equivalent and one can be produced from the other by a conjugation with $e^{i \CG}$, where $\CG$ is an almost local self-adjoint observable $g$-localized at $j=0$ and bounded $\|\CG\|\leq C$  for some $g(r)$ and $C$ which only depend on $f(r)$.
\end{lemma}

\begin{proof}

Unitary equivalence of states $\psi$ and $\psi_0$ follows from Corollary 2.6.11 of \cite{bratteli2012operator}. 
Let $(\Pi_0,\CH_0,|0\ral)$ be the GNS data for $\psi_0$. The state $\psi$ is a vector state corresponding to $|\psi \ral \in \CH_0$. Let $\CV_{n}$ be a subspace of $\CH_0$ spanned by vectors which can be produced from $|0\ral$ by an observable localized on $\Gamma_n$. Note that $\CV_1 \subset \CV_2 \subset \CV_3 \subset ...$.

Let $n_0$ be such that $f(n_0)<1/2$. Let us temporarily fix $n\geq n_0$ and not indicate it explicitly. Let us estimate the angle between the vector $|\psi\ral$ and the subspace $\CV=\CV_n$. The Hilbert space $\CH_0$ is isomorphic to $\CH_{\Gamma}\otimes\CH_{\bar\Gamma}$, where the Hilbert spaces $\CH_{\Gamma}$ and $\CH_{\bar{\Gamma}}$ carry representations of $\SA_{\Gamma}$ and $\SA_{\bar\Gamma}$, respectively. The restrictions of vector states $\psi$ and $\psi_0$ to $\SA_{\bar{\Gamma}}$ can be described by density matrices $\rho$ and $\rho_0$ on $\CH_{\bar\Gamma}$. The density matrix $\rho_0$ is pure, but $\rho$ is mixed, in general. We have
\beq
\|\rho-\rho_0\|_1 \leq \eps
\eeq
where $\eps=f(n)\in (0,1)$. Fuchs–van de Graaf inequality implies that for fidelity we have
\beq
F(\rho,\rho_0):= \| (\rho)^{1/2} (\rho_0)^{1/2} \|_{1} \geq 1-\frac{\eps}{2}.
\eeq

Let
\beq
|\psi \ral = \sum_{i=1}^N \sqrt{\lambda_i}\, |\eta_i\ral \otimes \, |\xi_i\ral
\eeq
be the Schmidt decomposition of $|\psi \ral$. Here $N\leq {\rm dim}\CH_\Gamma$,  $|\eta_i\ral$, $i=1,\ldots,N,$ are orthonormal vectors in $\CH_\Gamma$,  $|\xi_i\ral$, $i=1,\ldots,N,$ are orthonormal vectors in  $\CH_{\bar\Gamma}$, and $\lambda_i$, $i=1,\ldots,N,$ are positive numbers satisfying $\sum_i\lambda_i=1$. Since $|0\ral$ is factorized, its Schmidt decomposition contains only a single term:
\begin{equation}
|0\ral=|0_\Gamma\ral\otimes |0_{\bar\Gamma}\ral
\end{equation}

Let $a_i=\lal \xi_i |0_{\bar\Gamma}\ral$.
The fidelity of $\rho$ and $\rho_0$ can be expressed in terms of $\lambda_i$ and $a_i$:
\beq
F(\rho,\rho_0)=\left(\sum_i \lambda_i |a_i|^2\right)^{1/2}.
\eeq
 We define $|v\ral=\sum_i 
 \sqrt\lambda_i a^*_i |\eta_i\ral$ and let 
\begin{equation}
 |\chi\ral=\left(\sum_k\lambda_k  |a_k|^2\right)^{-1/2}|v\ral\otimes |0_{\bar\Gamma}\ral
\end{equation}
Then it is easy to see that
\begin{equation}\label{angleestimate}
|\lal\psi|\chi\ral|=\left(\sum_i \lambda_i |a_i|^2\right)^{1/2}=F(\rho,\rho_0)\geq 1-\frac{\eps}{2}.
\end{equation}
Since $\eps<1/2$, $|\psi\ral$ is not orthogonal to the subspace $\CV$.

For any $n\geq n_0$ let $|\chi_n\ral\in\CV_n$ be as above (geometrically, it is the normalized projection of $|\psi\ral$ to $\CV_n$). The estimate (\ref{angleestimate}) implies
\beq
|\lal \chi_{n} | \chi_{n+1} \ral| \geq 1-2 \eps_n,
\eeq
where $\eps_n=f(n)$. Let $\CU_{n_0}=e^{i \CG_{n_0}}$ be a unitary localized on $\Gamma_{n_0}$ that implements a rotation of $| 0\ral$ to $| \chi_{n_0} \ral$ with $\| \CG_{n_0} \| \leq \pi$. We can also choose unitary observables $\CU_n$ for $n\geq n_0$ localized on $\Gamma_{n+1}$ and satisfying $\|1 - \CU_n\| \leq (4 \eps_n)^{1/2}$ which implement rotations of $| \chi_n \ral$ to $| \chi_{n+1} \ral$, and which therefore can be written as $\CU_n = e^{i \CG_n}$ for an observable $\CG_n$ local on $\Gamma_{n+1}$ with $\| \CG_n \| \leq 2 (2 \eps_n)^{1/2}$. The ordered product of all such unitaries over $n \geq n_0$ can be written as $\CU=e^{i \CG}$. By construction, this unitary  maps  $|0 \ral$ to $| \psi \ral$. Moreover, since $\| \CG_n \| \leq 2 (2 f(n))^{1/2}$ for $n\geq n_0$, $\CG$ is $g$-localized for some MDP function $g(r)=\Or$ that only depends on $f(r)$, and $\|\CG\| \leq \sum_{n=n_0}^{\infty} 2 (2 f(n))^{1/2} + \pi$, a quantity that also depends only on $f(r)$.

\end{proof}

\begin{corollary} \label{cor:almfact}

Let $\psi_0$ and $\psi$ be distinct vector states on $\SA$ which satisfy the conditions of the above lemma. Let $\psi_s$, $s\in[0,1],$ be a path of vector states corresponding to a normalization of the path of vectors $s |0\ral + (1-s)|\psi\ral$. Then there exists a continuous path of self-adjoint almost local observables $\CG(s)$ $h$-localized at $j=0$ such that $\psi_{s} = \alpha_{\CG}(s)(\psi_0)$ and $\| \CG(s) \| \leq C$, for some MDP function $h(r)=\Or$ and $C>0$ which only depend on $f$.
\end{corollary}
\begin{proof}

By the above lemma  $|\psi\ral = \Pi_{0}(\CU)|0\ral$ for some $g$-localized $\CU$. Therefore the states $\psi_{s}$ are all $g$-close to the state $\psi_0$. The vectors $|\chi_{n,s}\ral$ depend continuously on $s$, therefore we can choose the unitaries $\CU_{n,s}$ and the observables $\CG_{n,s}$ so that they are continuous functions of $s$. Let $\tilde{\CU}_{n,s}$ be a product $\CU_{1,s} \cdots \CU_{n,s}$ generated by an almost local observable $\tilde{\CG}_{n,s}$. It follows from eq. (\ref{eq:LGAcomposition}) that $\|\tilde{\CG}_{n+1,s} - \tilde{\CG}_{n,s} \| \leq 2 (2 \eps_n)^{1/2}$. Therefore, the limit $\CG(s) = \lim_{n \to \infty} \tilde{\CG}_{n,s}$ is a continuous function of $s$.
\end{proof}

\begin{corollary} \label{cor:almcstatesSRE}
Lemma \ref{lma:almcstates} and Corollary \ref{cor:almfact} hold if we replace $\psi_0$ by an SRE state $\phi$ and $\psi$ by a state $\tilde{\phi}$ which is $f$-close to $\phi$. 
\end{corollary}

\begin{proof}
Let $\alpha_F$ be an LGA, such that $\phi \circ \alpha_F = \psi_0$. The state $\tilde{\phi} \circ \alpha_F$ is $g$-close to $\psi_0$ for some $g(r) = \Or$ that depends on $f(r)$ only. Both Lemma \ref{lma:almcstates} and Corollary \ref{cor:almfact} hold for $\phi \circ \alpha_F$ and  $\tilde{\phi} \circ \alpha_F$, and therefore they both hold for $\phi$ and $\tilde{\phi}$.
\end{proof}




\subsection{A classification of invertible phases without symmetries}

Short-Range Entangled states are invertible by definition. In this section we show that the converse is also true,\footnote{This is not true in the case of fermionic systems: Kitaev chain provides a counter-example.} and thus all invertible phases of bosonic 1d systems without symmetries are trivial.

We will say that a pure 1d state has bounded entanglement entropy if the entanglement entropies of all intervals $[j,k]\subset\ZZ$ are uniformly bounded.
\begin{remark}
It was shown by Matsui \cite{Matsui} that if $\psi$ has bounded entanglement entropy then it has the split property: the von Neumann algebras $\CM_{A}=\Pi_\psi(\SA_A)''$ and $\CM_{\bar A}=\Pi_\psi(\SA_{\bar A})''$ for a half-line $A$ are Type I von Neumann algebras. Since they are each other's commutants and generate $B(H_\psi)$, they must be Type I factors. Thus $\CM_A\simeq B(\CH_A)$ for some Hilbert space $\CH_A$, and the restriction $\psi|_A$ to a half-line $A$ can be described by a density matrix $\rho_A$ on $\CH_A$. In fact, \cite{Matsui} shows that $\CH_A$ can be identified with the GNS Hilbert space of one of the Schmidt vector states of $\psi$, which are all unitarily equivalent.
\end{remark}
\begin{lemma}\label{lma:bounded}
Both SRE 1d states and invertible 1d states have bounded entanglement entropy. 
\end{lemma}

\begin{proof}
Suppose we have a state $\psi$ obtained from a factorized pure state $\psi_0$ by conjugation with almost local unitaries $\CU_j$ and $\CU_k$ which are  $g$-localized at sites $j$ and $k$, correspondingly. Since conjugation by $\CU_{j,k}$ is an automorphism of $\SA$ which is almost localized on $j,k$, for any $\CA_l \in \SA_l$ we have 
\beq
|\psi_0(\CU_j \CU_k \CA_l \CU_k^* \CU_j^*) - \psi_0(\CA_l)| \leq \left(g(|j-l|) + g(|k-l|)\right)\|\CA_l\|.
\eeq
By Fannes' inequality \cite{Fannes}, the entropy of the site $l$ in the state $\psi$ is bounded by $h(|j-l|)+h(|k-l|)$ for some MDP function $h(r)=\Or$ that depends only on $g(r)$ and the asymptotics of $d_j$ for $j\ra\pm\infty$. Therefore such a state has a uniform bound on the  entanglement entropy of any interval $[j,k]$. The decomposition eq. (\ref{eq:autdecomposition2}) then implies that the same is true for any SRE or invertible state.
\end{proof}

Let $\psi$ be a possibly mixed state on a half-line $A$ which is $f$-close to a pure factorized state $\psi_0$ on $A$. By Corollary 2.6.11 of  \cite{bratteli2012operator}, $\psi$ is normal in the GNS representation of $\psi_0$ and can be described by a density matrix.

\begin{lemma}\label{lma:aux2}
Let $\psi$ be a pure 1d state on $\SA$. Suppose there is an $R>0$ such that $\psi$ is $f$-close far from $j=0$ to a pure factorized state $\omega^+$ on $(R,+\infty)$ and is $f$-close far from $j=0$ to a pure factorized state $\omega^-$ on $(-\infty,-R)$. Then it is $g$-close on $\ZZ$ far from $j=0$ to $\omega^+\otimes\omega^-$ for some MDP function $g(r)=\Or$ which depends only on $f$ and the asymptotics of $d_j$ for $j\ra \pm\infty$.
\end{lemma}

\begin{proof}
Since the states are split, we can describe them using density matrices on appropriate Hilbert spaces. The decomposition of $\ZZ$ into the union $(-\infty,-n)\sqcup \Gamma_n\sqcup (n,+\infty)$ gives rise to a tensor product decomposition $\CH=\CH^-_n\otimes\CH_{\Gamma_n}\otimes\CH^+_n$. Let $\psi_n^+$ and $\omega_n^{+}$ be restrictions of $\psi$ and $\omega^{+}$ to $(n,+\infty)$, and $\psi^-_n$ and $\omega^-_n$ be restrictions of $\psi$ and $\omega^-$ to $(-\infty,-n)$. Let $\rho^\pm_n$ and $\sigma^\pm_n$ be the corresponding density matrices. Let $\psi_n$ be a restriction of $\psi$ to $(-\infty,-n) \cup (n, + \infty)$ with the corresponding density matrix $\rho_n$. For $n>R$ we have
\beq
\| \rho_n^{\pm} - \sigma_n^{\pm} \|_{1} \leq f(n).
\eeq
Since trace norm is multiplicative under tensor product, we have
\beq \label{ineq:rhosigma}
\|(\rho_n^- \otimes \rho_n^+)-(\sigma_n^- \otimes \sigma_n^+)\|_{1} \leq \| \rho_n^{-} - \sigma_n^{-} \|_{1} + \| \rho_n^{+} - \sigma_n^{+} \|_{1} \leq 2 f(n)
\eeq
On the other hand, Fannes' inequality implies that for sufficiently large $n$ the entropy of $\rho^{\pm}_n$ is upper-bounded by MDP function $h(n)=\On$, where $h(n)$ depends only on $f(n)$ and the asymptotics of $d_j$ for $j\ra +\infty$. Therefore mutual informations $I(\rho^-_n:\rho^+_n)$ are also upper-bounded by $h(n)$, and the quantum Pinsker inequality implies
\beq
\|\rho_n-(\rho_n^- \otimes \rho_n^+)\|_{1}\leq 2 \sqrt {h(n)}.
\eeq
Combining this with eq. (\ref{ineq:rhosigma}), we get
\beq
\|\rho_n-(\sigma_n^- \otimes \sigma_n^+)\|_{1}\leq  2 f(n) + 2 \sqrt{h(n)}.
\eeq
\end{proof}

We say that a set of (ordered) eigenvalues $\{\lambda_j\}$ has {\it $g(r)$-decay} if $\eps(k) \leq g(\log(k))$ for some MDP function $g(r)= \Or$, where $\eps(k) = \sum_{j=k+1}^{\infty} \lambda_j$.

\begin{lemma} \label{lma:schmidt}
Let $\psi$ be a state on a half-line $A$ which is $f$-close far from the origin of $A$ to a pure factorized state $\psi_0$. Then its density matrix (in the GNS Hilbert space of this factorized state) has eigenvalues with $g(r)$-decay for some $g(r)=\Or$ that depends only on $f(r)$ and the  asymptotic behavior of $d_j={\rm dim}\,\CV_j$ for $j\ra\infty$. Conversely, for any density matrix on a half-line $A$ (in the GNS Hilbert space of a pure factorized state) whose eigenvalues have $g(r)$-decay there is a state on that half-line which has the same eigenvalues and is $f$-close far from the origin of $A$ to this pure factorized state. Furthermore, one can choose $f(r)$ so that it depends only on $g(r)$ and the  asymptotic behavior of $d_j={\rm dim}\,\CV_j$ for $j\ra \infty$. 
\end{lemma}

\begin{proof}
Suppose $\psi$ is $f$-close to a pure factorized state far from the origin. It can be purified on the whole line (e.g. in a system consisting of the given system on a half-line and its reflected copy on the other half-line). Moreover, by Lemma \ref{lma:aux2} we can choose this pure state to be $f'$-close far from the origin to a pure factorized state on the whole line for some MDP function $f'(r)=\Or$ that depends only on $f$. By Lemma \ref{lma:almcstates}, it can be produced from a pure factorized state on the whole line by a unitary observable $\CU$ which is $h$-localized for some $h$ which depends only on $f'$. Let $|\psi_0\ral$ be a GNS vector for the corresponding factorized state $\psi_0$. By Lemma A.1 of \cite{ThoulessHall},  there is an MDP function $h'(r)=\Or$ such that for any $r>0$ there is a unitary observable $\CU^{(r)}$ localized on a disk $\Gamma_r$ of radius $r$ such that
\beq\label{ineq1}
\| \Pi_{\psi_0}(\CU)|\psi_0\ral - \Pi_{\psi_0}(\CU^{(r)})|\psi_0\ral  \| \leq h'(r).
\eeq
On the other hand we have
\beq\label{ineq2}
\| \Pi_{\psi_0}(\CU)|\psi_0\ral - \Pi_{\psi_0}(\CU^{(r)})|\psi_0\ral  \| \geq \| \rho - \rho^{(r)} \|_1
\eeq
where $\rho$ is the density matrix for $\psi$ and $\rho^{(r)}$ is the density matrix for $\Pi_{\psi_0}(\CU^{(r)})|\psi_0\ral$ on $A$. The tracial distance between any two density matrices $\rho$ and $\rho'$ can be bounded from below in terms of their eigenvalues \cite{BSimon}:
\beq
\|\rho-\rho'\|_{1}\geq \sum_{j=1}^\infty |\lambda_j(\rho)-\lambda_j(\rho')|,
\eeq
where the eigenvalues $\lambda_i$ are ordered in decreasing order. Applying this to $\rho$ and $\rho^{(r)}$ and noting that $\rho^{(r)}$ has rank at most $\dim \CH_{A \cap \Gamma_r}$, we get
\beq\label{ineq3}
\| \rho - \rho^{(r)} \|_{1}\geq \eps(\dim \CH_{A \cap \Gamma_r}).
\eeq
Combining (\ref{ineq1}), (\ref{ineq2}) and (\ref{ineq3}) we get
\beq
\eps(\dim \CH_{A \cap \Gamma_r})\leq h'(r).
\eeq
Since $\dim \CH_{A \cap \Gamma_r}$ is upper-bounded by $\exp(c r^\alpha)$ for some positive constants $c$ and $\alpha$, we have $\eps(k) \leq g(\log(k))$ for some $g(r) = \Or$ which depends only on $f(r)$, $c$ and $\alpha$.

Conversely, suppose we are given a density matrix on a half-line $A$ with eigenvalues $\lambda_1\geq\lambda_2\geq...$. We may assume that the dimension of $\SA_A$ is infinite, since otherwise the statement is obviously true. Pick any pure factorized state $\psi_0$ on $A$ and choose a basis in each on-site Hilbert space $\CV_j$, $j\in A,$ such that for all $j$ the first basis vector  gives the state $\psi_0\vert_{\SA_j}$. This gives a lexicographic basis $|n\ral$, $n\in\mathbb N$, in the GNS Hilbert space of $\psi_0$. By our assumption on the growth of dimensions of $d_j={\rm dim}\  \CV_j$, there are positive constants $c$ and $\alpha$ such that for any $r$ and any $n<e^{c r^\alpha}$ the vector state $|n\ral \lal n|$ coincides with $\psi_0$ outside of  $\Gamma_r$. Therefore the state $\sum_{n=1}^\infty \lambda_n |n \ral \lal n|$ is $f$-close to $\psi_0$, where  $f(r)=g(c r^\alpha)=\Or$.
\end{proof}


By Lemma \ref{lma:schmidt} any restriction $\psi|_A$ of a state $\psi$ with $g(r)$-decay of Schmidt coefficients to a half-line $A$ can be purified by a state on $\bar{A}$, which is $f$-close far from the origin of $A$ to a pure factorized state for some $f(r)$ that depends on $g(r)$ only. We call such state a {\it truncation} of $\psi$ to $A$. Clearly, if $\omega$ is a truncation of $\psi$ to $A$, then $\omega|_A = \psi|_A$. The following lemma shows that truncations exist for all  invertible states.

\begin{lemma}  \label{lma:invSchmidt}
Let $\psi$ be an invertible state with an inverse $\psi'$, such that that $\Psi = \psi \otimes \psi'$ can be produced by an $f$-local LGA $(\alpha_F)^{-1}$ from a pure factorized state $\Psi_0 = \psi_0 \otimes \psi_0'$. Then $\psi$ has $g(r)$-decay of Schmidt coefficients for some $g(r)=\Or$ that depends only on $f(r)$. 
\end{lemma}

\begin{proof}
By Lemma  \ref{lma:bounded} and the results of \cite{Matsui}, $\Psi$, $\psi$, and $\psi'$ have the split property. Therefore for any half-line $A\subset\ZZ$ the corresponding GNS Hilbert spaces factorize into Hilbert spaces for $A$ and Hilbert spaces for $\bar A$, and the restrictions $\psi|_A$, $\psi'|_A$, $\Psi|_{A}$ can be described by density matrices $\rho_A$, $\rho'_A$ and $P_A$ in the Hilbert spaces for $A$ \cite{Matsui}. The restriction of $\Psi \circ \alpha_{F|_{\bar{A}}}$ on $\bar{A}$ is $f$-close far from the origin of $A$ to a pure factorized state, and therefore by Lemma \ref{lma:schmidt} the density matrix $P_A$ has $g(r)$-decay of Schmidt coefficients for some MDP function $g(r)=\Or$ that depends on $f$. Since $P_A = \rho_A \otimes \rho'_A$, the same is true for $\rho_A$ and $\rho'_A$. 

\end{proof}

\begin{lemma} \label{lma:12invertibleSRE}
Any truncation of an invertible state $\psi$ to any half-line is in a trivial phase.
\end{lemma}

\begin{proof}
Let $\psi$ be an invertible state on $\SA$ with an inverse $\psi'$ such that that $\Psi = \psi \otimes \psi'$ can be produced by an $f$-local LGA $(\alpha_F)^{-1}$ from a pure factorized state $\Psi_0 = \psi_0 \otimes \psi_0'$. Here $f(r)=\Or$ is an MDP function.

Let $\phi_k$ be a truncation of $\psi$ to $[k,\infty)$, and let $\phi'_k$ be a truncation of $\psi'$ to $[k,\infty)$. Since $(\phi_k \otimes \phi'_k) \circ \alpha_{F|_{[k,\infty)}}$ is $g$-close to $\psi_0 \otimes \psi'_0$ for some $g(r)=\Or$ that depends on $f(r)$ only, Lemma \ref{lma:almcstates} implies that $\phi_k$ is invertible with the inverse $\phi'_k$. Let $\tilde{\phi}_k$ be a pure state on $\SA^{(1)} \otimes \SA^{(2)}$, where $\SA^{(1)}$ and $\SA^{(2)}$ are two copies of $\SA$, with the following two properties: (1) its restriction to  $(-\infty,k)$ coincides with the factorized pure state $(\psi_0 \otimes \psi_0)|_{(-\infty,k)}$; (2) its restriction to $A=[k,+\infty)$ is a purification of $\psi|_{A}$ on $\SA^{(1)}_A$ by some state on $\SA^{(2)}_A$ which is $f'$-close to a factorized state for some $f'(r)=\Or$ that depends on $f$ only. The existence of such a state follows from Lemma \ref{lma:invSchmidt}. Similarly, we can define a state $\tilde{\phi}'_k$ on $\SA'^{(1)} \otimes \SA'^{(2)}$ which is an inverse of $\tilde{\phi}_k$. We let $\alpha_{\tilde{G}}$ be an LGA that maps  $\tilde{\phi}_k \otimes \tilde{\phi}'_k$ to   $\psi_0 \otimes \psi_0 \otimes \psi'_0 \otimes \psi'_0$. This LGA can be chosen to be the identity on $(-\infty,k)$ and $g'$-local with $g'(r)=\Or$ depending on $g(r)$ only.

Let us first show that the state $\psi_0 \otimes \phi_k \otimes \psi_0' \otimes \psi_0$ can be transformed into $\tilde{\phi}_k \otimes \psi_0' \otimes \psi_0$ by applying a certain LGA $\beta$, then an LGA $h_2$-localized at $k$ (for some MDP function $h_2(r)=\Or$), and finally the inverse of $\beta$. The sequence of steps is shown schematically in Fig. 1, where it is also indicated that $\beta$ is a composition of two LGAs described in more detail below. 

Equivalently, we can apply $\beta$ to both states and then show that the resulting states are related by an LGA $h_2$-localized at $k$. $\beta$ is a composition of two LGAs. The first one has the form $\Id \otimes \Id\otimes \alpha_G$, where $\alpha_G$ maps $\psi'_0\otimes\psi_0$ to $\phi'_k\otimes\phi_k$, see Fig. 1. The second one has the form $\Id \otimes \alpha_{F|_{[k+1,\infty)}} \otimes \Id$. The product of these two LGAs maps the two states of interest to the states $\Xi_k\otimes\phi_k$ and $\tilde{\Xi}_k \otimes\phi_k$, where both $\Xi_k$ and $\tilde{\Xi}$ are $h_1$-close on $\ZZ$ far from $k$ to the same pure factorized state on $\SA \otimes \SA\otimes\SA'$. Here $h_1(r)=\Or$ is an MDP function which depends only on $f(r)$. By Lemma \ref{lma:almcstates} $\Xi_k$ and $\tilde{\Xi}_{k}$ are related by a conjugation with a  unitary observable which is $h_2$-localized at $k$ for some $h_2(r)=\Or$ depending only on $f(r)$. 

Let us fix $L\in \NN$. Since the state $\tilde{\phi}'_k$ on $\SA^{(1)}\otimes\SA^{(2)}$ is invertible, by Lemma \ref{lma:invSchmidt} its restriction to $(-\infty,k+L)$ has $\tilde{g}(r)$-decay of Schmidt coefficients with $\tilde{g}(r)$ depending only on $f(r)$. Also, since the restriction of the state $\tilde{\phi}'_k$ to $(-\infty,k)$ is factorized, the nonzero Schmidt coeffients are the same as for the state $\tilde{\phi}_k|_{[k,k+L)}$. In particular, the number of nonzero Schmidt coefficients does not exceed the dimension of $\SA_{[k,k+L)}\otimes\SA_{[k,k+L)}$.  

Let us tensor  $\SA^{(1)}\otimes\SA^{(2)}$ with $\SA^{(3)} \otimes \SA^{(4)}$. By the proof of Lemma \ref{lma:schmidt} one can find a state $\chi_k$ on $\SA^{(1)}_{[k,k+L)}\otimes\SA^{(2)}_{[k,k+L)}\otimes\SA^{(3)}_{[k,k+L)}\otimes\SA^{(4)}_{[k,k+L)}$ which is a purification of $\tilde{\phi}_k|_{[k,k+L)}$ on $\SA^{(1)}_{[k,k+L)}\otimes\SA^{(2)}_{[k,k+L)}$ 
and such that its restriction to $\SA^{(3)}_{[k,k+L)} \otimes\SA^{(4)}_{[k,k+L)}$ is $h$-close to a factorized state far from $k+L$ for some $h(r)=\Or$ that depends on $f(r)$ only. The state $\chi_k$ can be produced from $(\psi_0 \otimes \psi_0 \otimes \psi_0 \otimes \psi_0)|_{[k,k+L)}$ by a conjugation with a unitary $\CV^{(0)}\in\SA^{(1)}_{[k,k+L)}\otimes\SA^{(2)}_{[k,k+L)}\otimes\SA^{(3)}_{[k,k+L)}\otimes\SA^{(4)}_{[k,k+L)}$.

Consider the states $\psi_0|_{[k,\infty)} \otimes \psi_0|_{[k,\infty)} \otimes \tilde{\phi_k}|_{[k,\infty)}$ and $\chi_k \otimes (\psi_0|_{[k+L,\infty)} \otimes \psi_0|_{[k+L,\infty)} \otimes \tilde{\phi}_{k+L})$ on the algebra $\SA^{(3)}_{[k,\infty)}\otimes \SA^{(4)}_{[k,\infty)}\otimes \SA^{(1)}_{[k,\infty)}\otimes \SA^{(2)}_{[k,\infty)}$ (see Fig. \ref{fig:Lemma46}). They are stably related by an almost local unitary which is $h'$-localized at $k+L$, with $h'(r)=\Or$ which depends only on $f(r)$. Indeed, we can first tensor both states with $\psi'_0 \otimes \psi'_0 \otimes \psi_0 \otimes \psi_0$ restricted to $[k,+\infty)$, then produce on these ancillas the state $\tilde{\phi}'_k \otimes \tilde{\phi}_k$ with a $g'$-local LGA, and apply $\alpha_{\tilde{G}|_{(-\infty,k+L)}}\circ \alpha_{\tilde{G}|_{[k+L,\infty)}}$ acting on the tensor product of the original states and $\tilde{\phi}'_k$ in an obvious way. In the same way as in the previous paragraph one can argue that these states are related by an almost local at $k+L$ unitary. This implies that the original states are also related by almost local at $k+L$ unitary $\CU^{(0)}$. 

We have shown that the states $\tilde{\phi}_k$ and $\tilde{\phi}_{k+L}$ are related (after tensoring with a total of six copies of factorized states $\psi_0$ and $\psi'_0$) by a conjugation with an almost local unitary $\CU^{(0)}$ followed by a conjugation with a strictly local unitary $\CV^{(0)}$. Similarly, we can construct such unitaries $\CU^{(n)}$, $\CV^{(n)}$ relating stabilizations of $\tilde{\phi}_{k+nL}$ and $\tilde{\phi}_{k+(n+1)L}$. By Lemma \ref{lma:1dLGA} we can choose $L$ such that an ordered product of conjugations with $\prod_{n=0}^{\infty} \CU^{(n)}$ is an LGA. Since $\CV^{(n)}$ commute with $\CU^{(n')}$ for $n<n'$, an ordered product $\prod_{n=0}^{\infty} \CV^{(n)} \CU^{(n)}$ is equal to $\prod_{m=0}^\infty \CV^{(m)}\prod_{n=0}^\infty\CU^{(n)}$ and therefore is also an LGA. By construction it relates $\tilde{\phi}_k$ to a factorized state, and therefore $\phi_k$ is in the trivial phase.

\end{proof}

\begin{theorem} \label{thm:invertibleSRE}
Any invertible bosonic 1d state $\psi$ is in a trivial phase.
\end{theorem}

\begin{proof}
Let $\psi'$ be an inverse state for $\psi$, and let $(\alpha_F)^{-1}$ be an $f$-local LGA that produces $(\psi \otimes \psi')$ from a pure factorized state $(\psi_0 \otimes \psi_0')$. It enough to show that the state $\psi_0 \otimes \psi \otimes \psi_0' \otimes \psi_0$ on $\SA^{(1)} \otimes \SA^{(2)} \otimes \SA'^{(3)} \otimes \SA^{(4)}$ is stably SRE.

Let $A$ be a half-line $[0,\infty)$. Since $\SA^{(1)}_A$ is identical to $\SA^{(2)}_A$ there is a pure state on $\SA^{(1)}_{\bar{A}} \otimes \SA^{(2)}_{A}$ identical to $\psi$. Let $\omega_-$ its truncation to $\bar{A}$. Similarly, let $\omega_+$ be a truncation to $A$ of a pure state on $\SA^{(2)}_{\bar{A}} \otimes \SA^{(1)}_A$ identical to $\psi$. By Lemma \ref{lma:12invertibleSRE}, the state $(\omega_- \otimes \omega_+) \otimes \psi_0' \otimes \psi_0$ is stably SRE. Therefore it is enough to show that $(\omega_- \otimes \omega_+) \otimes \psi_0' \otimes \psi_0$ and $\psi_0 \otimes \psi \otimes \psi_0' \otimes \psi_0$ are related by an LGA. In fact,  as we show below, such an LGA can be generated by an almost local observable. 

First we apply to both states an LGA which acts only on the last two factors and produces $\psi' \otimes \psi$ out of $\psi'_0\otimes\psi_0$. This gives us $\psi_0 \otimes \psi \otimes \psi' \otimes \psi$ and $(\omega_- \otimes \omega_+) \otimes \psi' \otimes \psi$.


Second, we apply a composition of $\alpha_{F|_{\bar{A}}}$ and $\alpha_{F|_{A}}$ on $\SA^{(2)} \otimes \SA'^{(3)}$ to states $(\omega_- \otimes \omega_+) \otimes \psi'$ and $\psi_0\otimes\psi\otimes\psi'$ on $\SA^{(1)}\otimes\SA^{(2)}\otimes\SA'^{(3)}$.  This transformation is shown schematically in Fig. \ref{fig:TNstate}. By Lemma \ref{lma:aux2} this gives two states which are both  $g$-close far from $0$ to a pure factorized state $\psi_0\otimes\psi_0\otimes\psi'_0$ on $\SA^{(1)}\otimes\SA^{(2)}\otimes\SA'^{(3)}$, for some MDP function $g(r)=\Or$. By Lemma \ref{lma:almcstates} these states are related by a conjugation with an almost local unitary. Applying the first two steps backwards we conclude that the two  original states are related by an LGA generated by an almost local observable.


\end{proof}

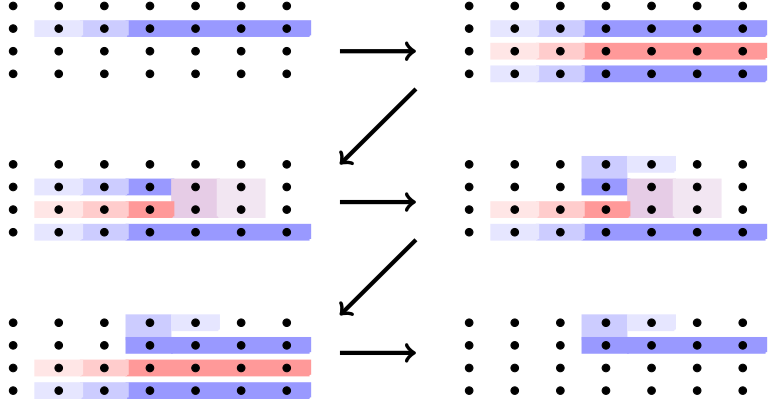
\begin{figure}
\centering

\begin{tikzpicture}

\begin{scope}[scale=.3,yshift=0cm,xshift=-20cm]

\draw [color=blue!40, fill=blue!40, very thick] (7,0.7) to (-1,0.7) -- (-1,1.3) -- (7,1.3) -- (7,0.7);

\draw [color=blue!20, fill=blue!20, very thick] (-1,0.7) to (-3,0.7) -- (-3,1.3) -- (-1,1.3) -- (-1,0.7);

\draw [color=blue!10, fill=blue!10, very thick] (-3,0.7) to (-5,0.7) -- (-5,1.3) -- (-3,1.3) -- (-3,0.7);

\foreach \x in {-6,-4,...,6}{                           
    \foreach \y in {-1,...,2}{                       
    \node[draw,black,circle,inner sep=1pt,fill] at (\x,\y) {}; 
    }
}
\end{scope}

\begin{scope}[scale=.3,yshift=0cm]

\draw [color=blue!40, fill=blue!40, very thick] (7,0.7) to (-1,0.7) -- (-1,1.3) -- (7,1.3) -- (7,0.7);

\draw [color=blue!20, fill=blue!20, very thick] (-1,0.7) to (-3,0.7) -- (-3,1.3) -- (-1,1.3) -- (-1,0.7);

\draw [color=blue!10, fill=blue!10, very thick] (-3,0.7) to (-5,0.7) -- (-5,1.3) -- (-3,1.3) -- (-3,0.7);

\draw [color=red!40, fill=red!40, very thick] (7,0.7-1) to (-1,0.7-1) -- (-1,1.3-1) -- (7,1.3-1) -- (7,0.7-1);

\draw [color=red!20, fill=red!20, very thick] (-1,0.7-1) to (-3,0.7-1) -- (-3,1.3-1) -- (-1,1.3-1) -- (-1,0.7-1);

\draw [color=red!10, fill=red!10, very thick] (-3,0.7-1) to (-5,0.7-1) -- (-5,1.3-1) -- (-3,1.3-1) -- (-3,0.7-1);

\draw [color=blue!40, fill=blue!40, very thick] (7,0.7-2) to (-1,0.7-2) -- (-1,1.3-2) -- (7,1.3-2) -- (7,0.7-2);

\draw [color=blue!20, fill=blue!20, very thick] (-1,0.7-2) to (-3,0.7-2) -- (-3,1.3-2) -- (-1,1.3-2) -- (-1,0.7-2);

\draw [color=blue!10, fill=blue!10, very thick] (-3,0.7-2) to (-5,0.7-2) -- (-5,1.3-2) -- (-3,1.3-2) -- (-3,0.7-2);

\foreach \x in {-6,-4,...,6}{                           
    \foreach \y in {-1,...,2}{                       
    \node[draw,black,circle,inner sep=1pt,fill] at (\x,\y) {}; 
    }
}
\end{scope}

\begin{scope}[scale=.3,yshift=-7cm,xshift=-20cm]

\draw [color=blue!40, fill=blue!40, very thick] (1,0.7) to (-1,0.7) -- (-1,1.3) -- (1,1.3) -- (1,0.7);

\draw [color=blue!20, fill=blue!20, very thick] (-1,0.7) to (-3,0.7) -- (-3,1.3) -- (-1,1.3) -- (-1,0.7);

\draw [color=blue!10, fill=blue!10, very thick] (-3,0.7) to (-5,0.7) -- (-5,1.3) -- (-3,1.3) -- (-3,0.7);

\draw [color=violet!20, fill=violet!20, very thick] (1,0.7-1) to (3,0.7-1) -- (3,1.3) -- (1,1.3) -- (1,0.7-1);

\draw [color=violet!10, fill=violet!10, very thick] (3,0.7-1) to (5,0.7-1) -- (5,1.3) -- (3,1.3) -- (3,0.7-1);

\draw [color=red!40, fill=red!40, very thick] (1,0.7-1) to (-1,0.7-1) -- (-1,1.3-1) -- (1,1.3-1) -- (1,0.7-1);

\draw [color=red!20, fill=red!20, very thick] (-1,0.7-1) to (-3,0.7-1) -- (-3,1.3-1) -- (-1,1.3-1) -- (-1,0.7-1);

\draw [color=red!10, fill=red!10, very thick] (-3,0.7-1) to (-5,0.7-1) -- (-5,1.3-1) -- (-3,1.3-1) -- (-3,0.7-1);

\draw [color=blue!40, fill=blue!40, very thick] (7,0.7-2) to (-1,0.7-2) -- (-1,1.3-2) -- (7,1.3-2) -- (7,0.7-2);

\draw [color=blue!20, fill=blue!20, very thick] (-1,0.7-2) to (-3,0.7-2) -- (-3,1.3-2) -- (-1,1.3-2) -- (-1,0.7-2);

\draw [color=blue!10, fill=blue!10, very thick] (-3,0.7-2) to (-5,0.7-2) -- (-5,1.3-2) -- (-3,1.3-2) -- (-3,0.7-2);

\foreach \x in {-6,-4,...,6}{                           
    \foreach \y in {-1,...,2}{                       
    \node[draw,black,circle,inner sep=1pt,fill] at (\x,\y) {}; 
    }
}
\end{scope}

\begin{scope}[scale=.3,yshift=-7cm,xshift=0cm]


\draw [color=blue!20, fill=blue!20, very thick] (1,1.3) to (-1,1.3) -- (-1,2.3) -- (0.87,2.3) -- (0.87,1.3);

\draw [color=blue!10, fill=blue!10, very thick] (3,1.7) to (1,1.7) -- (1,2.3) -- (3,2.3) -- (3,1.7);

\draw [color=blue!40, fill=blue!40, very thick] (1,0.7) to (-1,0.7) -- (-1,1.3) -- (1,1.3) -- (1,0.7);


\draw [color=violet!20, fill=violet!20, very thick] (1,0.7-1) to (3,0.7-1) -- (3,1.3) -- (1,1.3) -- (1,0.7-1);

\draw [color=violet!10, fill=violet!10, very thick] (3,0.7-1) to (5,0.7-1) -- (5,1.3) -- (3,1.3) -- (3,0.7-1);

\draw [color=red!40, fill=red!40, very thick] (1,0.7-1) to (-1,0.7-1) -- (-1,1.3-1) -- (1,1.3-1) -- (1,0.7-1);

\draw [color=red!20, fill=red!20, very thick] (-1,0.7-1) to (-3,0.7-1) -- (-3,1.3-1) -- (-1,1.3-1) -- (-1,0.7-1);

\draw [color=red!10, fill=red!10, very thick] (-3,0.7-1) to (-5,0.7-1) -- (-5,1.3-1) -- (-3,1.3-1) -- (-3,0.7-1);

\draw [color=blue!40, fill=blue!40, very thick] (7,0.7-2) to (-1,0.7-2) -- (-1,1.3-2) -- (7,1.3-2) -- (7,0.7-2);

\draw [color=blue!20, fill=blue!20, very thick] (-1,0.7-2) to (-3,0.7-2) -- (-3,1.3-2) -- (-1,1.3-2) -- (-1,0.7-2);

\draw [color=blue!10, fill=blue!10, very thick] (-3,0.7-2) to (-5,0.7-2) -- (-5,1.3-2) -- (-3,1.3-2) -- (-3,0.7-2);

\foreach \x in {-6,-4,...,6}{                           
    \foreach \y in {-1,...,2}{                       
    \node[draw,black,circle,inner sep=1pt,fill] at (\x,\y) {}; 
    }
}
\end{scope}

\begin{scope}[scale=.3,yshift=-14cm,xshift=-20cm]

\draw [color=blue!20, fill=blue!20, very thick] (1,1.3) to (-1,1.3) -- (-1,2.3) -- (0.87,2.3) -- (0.87,1.3);

\draw [color=blue!10, fill=blue!10, very thick] (3,1.7) to (1,1.7) -- (1,2.3) -- (3,2.3) -- (3,1.7);

\draw [color=blue!40, fill=blue!40, very thick] (1,0.7) to (-1,0.7) -- (-1,1.3) -- (1,1.3) -- (1,0.7);

\draw [color=blue!40, fill=blue!40, very thick] (7,0.7) to (1,0.7) -- (1,1.3) -- (7,1.3) -- (7,0.7);



\draw [color=red!40, fill=red!40, very thick] (7,0.7-1) to (-1,0.7-1) -- (-1,1.3-1) -- (7,1.3-1) -- (7,0.7-1);

\draw [color=red!20, fill=red!20, very thick] (-1,0.7-1) to (-3,0.7-1) -- (-3,1.3-1) -- (-1,1.3-1) -- (-1,0.7-1);

\draw [color=red!10, fill=red!10, very thick] (-3,0.7-1) to (-5,0.7-1) -- (-5,1.3-1) -- (-3,1.3-1) -- (-3,0.7-1);

\draw [color=blue!40, fill=blue!40, very thick] (7,0.7-2) to (-1,0.7-2) -- (-1,1.3-2) -- (7,1.3-2) -- (7,0.7-2);

\draw [color=blue!20, fill=blue!20, very thick] (-1,0.7-2) to (-3,0.7-2) -- (-3,1.3-2) -- (-1,1.3-2) -- (-1,0.7-2);

\draw [color=blue!10, fill=blue!10, very thick] (-3,0.7-2) to (-5,0.7-2) -- (-5,1.3-2) -- (-3,1.3-2) -- (-3,0.7-2);

\foreach \x in {-6,-4,...,6}{                           
    \foreach \y in {-1,...,2}{                       
    \node[draw,black,circle,inner sep=1pt,fill] at (\x,\y) {}; 
    }
}
\end{scope}

\begin{scope}[scale=.3,yshift=-14cm,xshift=0cm]

\draw [color=blue!20, fill=blue!20, very thick] (1,1.3) to (-1,1.3) -- (-1,2.3) -- (0.87,2.3) -- (0.87,1.3);

\draw [color=blue!10, fill=blue!10, very thick] (3,1.7) to (1,1.7) -- (1,2.3) -- (3,2.3) -- (3,1.7);

\draw [color=blue!40, fill=blue!40, very thick] (1,0.7) to (-1,0.7) -- (-1,1.3) -- (1,1.3) -- (1,0.7);

\draw [color=blue!40, fill=blue!40, very thick] (7,0.7) to (1,0.7) -- (1,1.3) -- (7,1.3) -- (7,0.7);









\foreach \x in {-6,-4,...,6}{                           
    \foreach \y in {-1,...,2}{                       
    \node[draw,black,circle,inner sep=1pt,fill] at (\x,\y) {}; 
    }
}
\end{scope}

\draw[ultra thick, ->] (-3.5,0) -- (-3.5+1,0);

\draw[ultra thick, ->] (-3.5+1,-0.5) -- (-3.5,-1.5);

\draw[ultra thick, ->] (-3.5,-2) -- (-3.5+1,-2);

\draw[ultra thick, ->] (-3.5+1,-0.5-2) -- (-3.5,-1.5-2);

\draw[ultra thick, ->] (-3.5,-2-2) -- (-3.5+1,-2-2);

\end{tikzpicture}

\caption{The sequence of steps that transforms $\psi_0 \otimes \phi_k \otimes \psi'_0 \otimes \psi_0$ into $\tilde{\phi}_{k} \otimes \psi'_0 \otimes \psi_0$. The regions shaded in blue denote the entangled parts of $\psi$, while the regions shaded in red denote the entangled parts of $\psi'$. The regions where the state is close to a factorized state are schematically indicated by a faded shading.
}
\label{fig:Lemma44}
\end{figure}

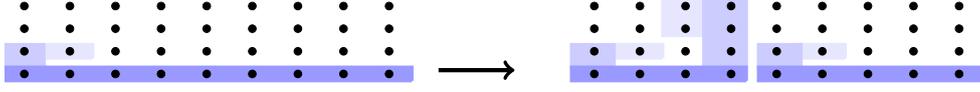
\begin{figure}
\centering

\begin{tikzpicture}

\begin{scope}[scale=.3,yshift=-7cm,xshift=-25cm]


\draw [color=blue!20, fill=blue!20, very thick] (0.87-4,1.3) to (-5+0.2,1.3) -- (-5+0.2,2.3) -- (0.87-4,2.3) -- (0.87-4,1.3);

\draw [color=blue!10, fill=blue!10, very thick] (-1,1.7) to (-3,1.7) -- (-3,2.3) -- (-1,2.3) -- (-1,1.7);

\draw [color=blue!40, fill=blue!40, very thick] (13,0.7) to (-5+0.2,0.7) -- (-5+0.2,1.3) -- (13,1.3) -- (13,0.7);










\foreach \x in {-4,-2,...,12}{                           
    \foreach \y in {1,...,4}{                       
    \node[draw,black,circle,inner sep=1pt,fill] at (\x,\y) {}; 
    }
}
\end{scope}

\begin{scope}[scale=.3,yshift=-7cm,xshift=0cm]


\draw [color=blue!10, fill=blue!10, very thick] (-1+10-2,1.7) to (-3+10-2,1.7) -- (-3+10-2,2.3) -- (-1+10-2,2.3) -- (-1+10-2,1.7);

\draw [color=blue!10, fill=blue!10, very thick] (-1,1.7) to (-3,1.7) -- (-3,2.3) -- (-1,2.3) -- (-1,1.7);

\draw [color=blue!10, fill=blue!10, very thick] (-1+2,2.7) to (-3+2,2.7) -- (-3+2,4.3) -- (-1+2,4.3) -- (-1+2,2.7);

\draw [color=blue!20, fill=blue!20, very thick] (0.87-4,1.3) to (-5,1.3) -- (-5,2.3) -- (0.87-4,2.3) -- (0.87-4,1.3);

\draw [color=blue!20, fill=blue!20, very thick] (0.87-4+6-0.2,1.3) to (-5+6-0.2,1.3) -- (-5+6-0.2,4.3) -- (0.87-4+6-0.2,4.3) -- (0.87-4+6-0.2,1.3);

\draw [color=blue!20, fill=blue!20, very thick] (0.87-4+10-2+0.2,1.3) to (-5+10-2+0.2,1.3) -- (-5+10-2+0.2,2.3) -- (0.87-4+10-2+0.2,2.3) -- (0.87-4+10-2+0.2,1.3);

\draw [color=blue!40, fill=blue!40, very thick] (3-0.13-0.2,0.7) to (-5,0.7) -- (-5,1.3) -- (3-0.13-0.2,1.3) -- (3-0.13-0.2,0.7);

\draw [color=blue!40, fill=blue!40, very thick] (13,0.7) to (5-2+0.2,0.7) -- (5-2+0.2,1.3) -- (13,1.3) -- (13,0.7);










\foreach \x in {-4,-2,...,12}{                           
    \foreach \y in {1,...,4}{                       
    \node[draw,black,circle,inner sep=1pt,fill] at (\x,\y) {}; 
    }
}
\end{scope}



\draw[ultra thick, ->] (-3.25,-1.75) -- (-3.25+1,-1.75);



\end{tikzpicture}

\caption{States $\tilde{\phi}_k$ and $\chi_{k} \otimes (\psi_0|_{k+L,\infty} \otimes \psi_0|_{k+L,\infty} \otimes \tilde{\phi}_{k+L})$
}
\label{fig:Lemma46}
\end{figure}

\begin{figure}  \label{fig:Theorem1}
\centering

\begin{tikzpicture}

\begin{scope}[scale=.3,xshift=-20cm]
\draw [color=blue!40, fill=blue!40, very thick] (-7,0.7) to (7,0.7) -- (7,1.3) -- (-7,1.3) -- (-7,0.7);

\foreach \x in {-6,-4,...,6}{                           
    \foreach \y in {-1,0,...,2}{                       
    \node[draw,black,circle,inner sep=1pt,fill] at (\x,\y) {}; 
    }
}
\end{scope}

\begin{scope}[scale=.3]
\draw [color=blue!40, fill=blue!40, very thick] (-7,0.7) to (7,0.7) -- (7,1.3) -- (-7,1.3) -- (-7,0.7);

\draw [color=red!40, fill=red!40, very thick] (-7,0.7-1) to (7,0.7-1) -- (7,1.3-1) -- (-7,1.3-1) -- (-7,0.7-1);

\draw [color=blue!40, fill=blue!40, very thick] (-7,0.7-2) to (7,0.7-2) -- (7,1.3-2) -- (-7,1.3-2) -- (-7,0.7-2);

\foreach \x in {-6,-4,...,6}{                           
    \foreach \y in {-1,0,...,2}{                       
    \node[draw,black,circle,inner sep=1pt,fill] at (\x,\y) {}; 
    }
}
\end{scope}

\begin{scope}[scale=.3,yshift=-7cm,xshift=-20cm]
\draw [color=blue!40, fill=blue!40, very thick] (-1,0.7) to (1,0.7) -- (1,1.3) -- (-1,1.3) -- (-1,0.7);

\draw [color=red!40, fill=red!40, very thick] (-1,0.7-1) to (1,0.7-1) -- (1,1.3-1) -- (-1,1.3-1) -- (-1,0.7-1);

\draw [color=violet!20, fill=violet!20, very thick] (-3,0.7-1) to (-1,0.7-1) -- (-1,1.3) -- (-3,1.3) -- (-3,0.7-1);

\draw [color=violet!10, fill=violet!10, very thick] (-5,0.7-1) to (-3,0.7-1) -- (-3,1.3) -- (-5,1.3) -- (-5,0.7-1);

\draw [color=violet!20, fill=violet!20, very thick] (1,0.7-1) to (3,0.7-1) -- (3,1.3) -- (1,1.3) -- (1,0.7-1);

\draw [color=violet!10, fill=violet!10, very thick] (3,0.7-1) to (5,0.7-1) -- (5,1.3) -- (3,1.3) -- (3,0.7-1);

\draw [color=blue!40, fill=blue!40, very thick] (-7,0.7-2) to (7,0.7-2) -- (7,1.3-2) -- (-7,1.3-2) -- (-7,0.7-2);

\foreach \x in {-6,-4,...,6}{                           
    \foreach \y in {-1,0,...,2}{                       
    \node[draw,black,circle,inner sep=1pt,fill] at (\x,\y) {}; 
    }
}
\end{scope}

\begin{scope}[scale=.3,yshift=-7cm]
\draw [color=blue!40, fill=blue!40, very thick] (-1,0.7) to (1,0.7+1) -- (1,1.3+1) -- (-1,1.3) -- (-1,0.7);

\draw [color=blue!40, fill=blue!40, very thick] (-1,0.7+1) to (1,0.7) -- (1,1.3) -- (-1,1.3+1) -- (-1,0.7+1);

\draw [color=red!40, fill=red!40, very thick] (-1,0.7-1) to (1,0.7-1) -- (1,1.3-1) -- (-1,1.3-1) -- (-1,0.7-1);

\draw [color=blue!20, fill=blue!20, very thick] (-3,0.7+1) to (-1,0.7+1) -- (-1,1.3+1) -- (-3,1.3+1) -- (-3,0.7+1);

\draw [color=blue!10, fill=blue!10, very thick] (-5,0.7+1) to (-3,0.7+1) -- (-3,1.3+1) -- (-5,1.3+1) -- (-5,0.7+1);

\draw [color=blue!20, fill=blue!20, very thick] (3,0.7+1) to (1,0.7+1) -- (1,1.3+1) -- (3,1.3+1) -- (3,0.7+1);

\draw [color=blue!10, fill=blue!10, very thick] (5,0.7+1) to (3,0.7+1) -- (3,1.3+1) -- (5,1.3+1) -- (5,0.7+1);

\draw [color=violet!20, fill=violet!20, very thick] (-3,0.7-1) to (-1,0.7-1) -- (-1,1.3) -- (-3,1.3) -- (-3,0.7-1);

\draw [color=violet!10, fill=violet!10, very thick] (-5,0.7-1) to (-3,0.7-1) -- (-3,1.3) -- (-5,1.3) -- (-5,0.7-1);

\draw [color=violet!20, fill=violet!20, very thick] (1,0.7-1) to (3,0.7-1) -- (3,1.3) -- (1,1.3) -- (1,0.7-1);

\draw [color=violet!10, fill=violet!10, very thick] (3,0.7-1) to (5,0.7-1) -- (5,1.3) -- (3,1.3) -- (3,0.7-1);

\draw [color=blue!40, fill=blue!40, very thick] (-7,0.7-2) to (7,0.7-2) -- (7,1.3-2) -- (-7,1.3-2) -- (-7,0.7-2);

\foreach \x in {-6,-4,...,6}{                           
    \foreach \y in {-1,0,...,2}{                       
    \node[draw,black,circle,inner sep=1pt,fill] at (\x,\y) {}; 
    }
}
\end{scope}

\begin{scope}[scale=.3,yshift=-14cm,xshift=-20cm]
\draw [color=blue!40, fill=blue!40, very thick] (-1,0.7) to (1,0.7+1) -- (1,1.3+1) -- (-1,1.3) -- (-1,0.7);

\draw [color=blue!40, fill=blue!40, very thick] (-1,0.7+1) to (1,0.7) -- (1,1.3) -- (-1,1.3+1) -- (-1,0.7+1);

\draw [color=blue!20, fill=blue!20, very thick] (-3,0.7+1) to (-1,0.7+1) -- (-1,1.3+1) -- (-3,1.3+1) -- (-3,0.7+1);

\draw [color=blue!10, fill=blue!10, very thick] (-5,0.7+1) to (-3,0.7+1) -- (-3,1.3+1) -- (-5,1.3+1) -- (-5,0.7+1);

\draw [color=blue!20, fill=blue!20, very thick] (3,0.7+1) to (1,0.7+1) -- (1,1.3+1) -- (3,1.3+1) -- (3,0.7+1);

\draw [color=blue!10, fill=blue!10, very thick] (5,0.7+1) to (3,0.7+1) -- (3,1.3+1) -- (5,1.3+1) -- (5,0.7+1);

\draw [color=blue!40, fill=blue!40, very thick] (-7,0.7) to (-1,0.7) -- (-1,1.3) -- (-7,1.3) -- (-7,0.7);

\draw [color=blue!40, fill=blue!40, very thick] (7,0.7) to (1,0.7) -- (1,1.3) -- (7,1.3) -- (7,0.7);

\draw [color=red!40, fill=red!40, very thick] (-7,0.7-1) to (7,0.7-1) -- (7,1.3-1) -- (-7,1.3-1) -- (-7,0.7-1);

\draw [color=blue!40, fill=blue!40, very thick] (-7,0.7-2) to (7,0.7-2) -- (7,1.3-2) -- (-7,1.3-2) -- (-7,0.7-2);

\foreach \x in {-6,-4,...,6}{                           
    \foreach \y in {-1,0,...,2}{                       
    \node[draw,black,circle,inner sep=1pt,fill] at (\x,\y) {}; 
    }
}
\end{scope}

\begin{scope}[scale=.3,yshift=-14cm]
\draw [color=blue!40, fill=blue!40, very thick] (-1,0.7) to (1,0.7+1) -- (1,1.3+1) -- (-1,1.3) -- (-1,0.7);

\draw [color=blue!40, fill=blue!40, very thick] (-1,0.7+1) to (1,0.7) -- (1,1.3) -- (-1,1.3+1) -- (-1,0.7+1);

\draw [color=blue!20, fill=blue!20, very thick] (-3,0.7+1) to (-1,0.7+1) -- (-1,1.3+1) -- (-3,1.3+1) -- (-3,0.7+1);

\draw [color=blue!10, fill=blue!10, very thick] (-5,0.7+1) to (-3,0.7+1) -- (-3,1.3+1) -- (-5,1.3+1) -- (-5,0.7+1);

\draw [color=blue!20, fill=blue!20, very thick] (3,0.7+1) to (1,0.7+1) -- (1,1.3+1) -- (3,1.3+1) -- (3,0.7+1);

\draw [color=blue!10, fill=blue!10, very thick] (5,0.7+1) to (3,0.7+1) -- (3,1.3+1) -- (5,1.3+1) -- (5,0.7+1);

\draw [color=blue!40, fill=blue!40, very thick] (-7,0.7) to (-1,0.7) -- (-1,1.3) -- (-7,1.3) -- (-7,0.7);

\draw [color=blue!40, fill=blue!40, very thick] (7,0.7) to (1,0.7) -- (1,1.3) -- (7,1.3) -- (7,0.7);



\foreach \x in {-6,-4,...,6}{                           
    \foreach \y in {-1,0,...,2}{                       
    \node[draw,black,circle,inner sep=1pt,fill] at (\x,\y) {}; 
    }
}
\end{scope}

\draw[ultra thick, ->] (-3.5,0) -- (-3.5+1,0);

\draw[ultra thick, ->] (-3.5+1,-0.5) -- (-3.5,-1.5);

\draw[ultra thick, ->] (-3.5,-2) -- (-3.5+1,-2);

\draw[ultra thick, ->] (-3.5+1,-0.5-2) -- (-3.5,-1.5-2);

\draw[ultra thick, ->] (-3.5,-2-2) -- (-3.5+1,-2-2);

\end{tikzpicture}

\caption{The sequence of steps that transforms $\psi_0 \otimes \psi \otimes \psi'_0 \otimes \psi_0$ into $(\omega_- \otimes \omega_+) \otimes \psi_0' \otimes \psi_0$. The coloring in the same as in Fig. \ref{fig:Lemma44}.
}
\label{fig:TNstate}
\end{figure}
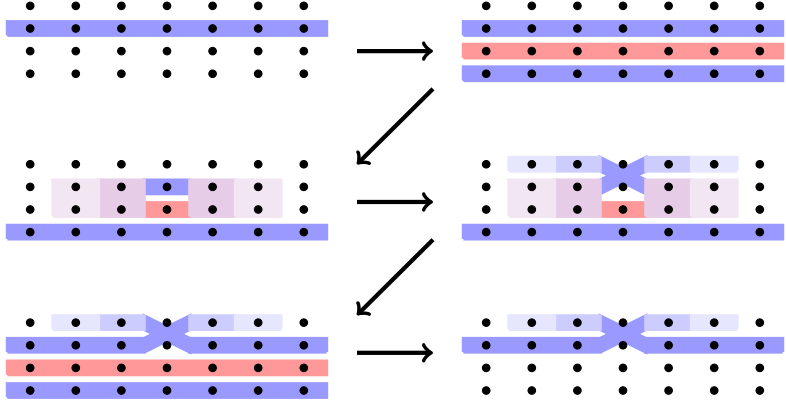

\subsection{A classification of invertible  phases with symmetries}

In this section we show that the index defined in Section 3 completely classifies invertible phases of 1d bosonic lattice systems with unitary symmetries.
\begin{theorem} \label{completeInvariance}
    A $G$-invariant invertible 1d state $\psi$ is in the trivial  $G$-invariant phase if and only if it has a trivial index.  
\end{theorem}
In order to prove the theorem, we need the following result from representation theory. Its  proof can be found in various sources such as theorem 4.4 in \cite{brocker2013representations}.
\begin{lemma} \label{lma:rep}
Let $V$ be a finite dimensional faithful representation of a finite group $G$. For any irreducible representation $W$ of $G$, $d(L)=\dim \Hom_G(W, V^{\otimes L})> 0$ for large enough $L$. Moreover, $d(L)$ grows exponentially with $L$.
\end{lemma}


\begin{proof} [Proof of theorem]

Let $V$ be a finite dimensional representation of $G$ whose subrepresentations contain every irreducible representation of $G$ (including the trivial one). Such a $V$ always exists, with $\mathbb C[G]$ being one example. Clearly $V$ is faithful. For all $j\in\ZZ$, let $\CV'_j=V$, $\SA'_j={\rm End}(\CV'_j)$, and let $\SA'$ be the norm-completion of $\otimes_j \SA'_j$.
Let $\psi'$ be a special $G$-invariant factorized pure state on $\SA'$ (it exists because $V$ contains the trivial representation).

Further, for any $j\in\Lambda$ we pick a one-dimensional representation $W_j$ of $G$ (to be fixed later) and let $\SA''_j$ be the norm-completion of $\otimes_j {\rm End}(W_j\oplus T_j)$, where $T_j\simeq\CC$ is the trivial representation of $G$. We also choose normalized basis vectors $w_j\in W_j$ and $t_j\in T_j$. If we ignore the $G$-action, then $\SA''$ corresponds to an infinite chain of qubits. We denote by $\psi''$ a $G$-invariant factorized pure state $\psi''$ on $\SA''$ whose restriction to $\SA''_j$ is given by $\langle w_j|\cdot | w_j\rangle $. Note that $\psi''$ is a special $G$-invariant factorized pure state if and only if all $W_j$ are trivial representations.

We will show that with appropriate choice of representations $W_j$ the state $\Psi = \psi \otimes \psi'\otimes\psi''$ on $\SA \otimes\SA'\otimes \SA''$ tensored with a finite number of copies of factorized pure states can be disentangled by a $G$-equivariant LGA (that is, can be mapped by a $G$-equivariant LGA to a $G$-invariant factorized pure state). Since, as shown in Appendix A, any $G$-invariant factorized pure state is $G$-stably equivalent to a special $G$-invariant factorized pure state, this implies that $\psi$ is in a trivial $G$-invariant phase.

By Lemma \ref{lma:bounded} and \cite{Matsui}, the state $\psi$ has the split property. Thus the von Neumann algebra 
$\CM_{> k}=\Pi_\psi(\SA_{> k})''$ on a half-line $(k,\infty)$ is a Type I factor. In other words, $\CM_{>k}$ is isomorphic to the algebra of bounded operators on a Hilbert space $\CW_{>k}$.
By Remark \ref{rem:Ug} and triviality of the index, the operators $U^g_{>k}$ define a unitary representation of $G$ on $\CW_{>k}$. The split property also implies that $\psi$ restricted to $(k, \infty)$ is a normal state with a density matrix $\rho_k$ (a positive operator on $\CW_{>k}$ with unit trace). 
As the restriction of $\psi$ to $\SA_{>k}$ is $G$-invariant, each eigenspace of $\rho_k$ is a representation of $G$. Therefore, each summand in the direct sum decomposition of $\CW_{>k}$ into irreducible (necessarily finite dimensional) representations of $G$ is spanned by eigenvectors of $\rho_k$ with equal eigenvalues.

By Lemma $\ref{lma:rep}$ each of these irreducible summands is contained in $V^{\otimes L}$ for large enough $L$ with exponentially growing multiplicity (with respect to $L$). Therefore in the same way as in Lemma \ref{lma:schmidt} we can construct a $G$-invariant state on $\SA_{> k}\otimes\SA'_{>k}$ which is $f$-close to a $G$-invariant factorized pure state, has the same eigenvalues as $\rho_k$, and the eigenspace for each eigenvalue transforms in the same representation of $G$ as the corresponding eigenspace of $\rho_k$. Here $f(r)=\Or$ is an MDP function which depends only on the localization of the LGA which produces $\psi\otimes\psi'$ from a pure factorized state and in particular is independent of $k$. Therefore there is a $G$-invariant pure state $\Xi^{(k)}$ on $\SA \otimes \SA'$ that coincides with $\psi\otimes\psi'$ on $(-\infty,k]$, while on $(k,+\infty)$ it is $f$-close far from $k$ to a $G$-invariant factorized pure state. In other words, $\Xi^{(k)}$ is a truncation of the  invertible state $\psi \otimes \psi'$ to $(-\infty,k]$. 

Similarly to the proof of Lemma \ref{lma:12invertibleSRE} we can also define $G$-invariant states $\tilde{\Xi}^{(k)}$ on $\SA^{(1)}\otimes \SA'^{(1)}\otimes\SA^{(2)}\otimes\SA'^{(2)}$ which have the following properties: (1) $\tilde{\Xi}^{(k)}$ is pure factorized when restricted to $(k,\infty)$, and (2) its restriction to $(-\infty,k]$ is a purification of $(\psi\otimes\psi')\vert_{(-\infty,k]}$ by a state on $\SA^{(2)}_{(-\infty,k]}\otimes\SA'^{(2)}_{(-\infty,k]}$ which is $g$-close to a pure factorized state for some $g(r)=\Or$ which depends only on $f$. In the following we assume that all states are tensored with a finite number of copies of factorized states, so that all stable equivalences correspond to equivalences without stabilization.

Let us fix $L\in\NN$. By the proof of Lemma \ref{lma:12invertibleSRE} the state $\tilde{\Xi}^{(k-L)}$ can be stably produced from $\tilde{\Xi}^{(k)}$ by first conjugation with an almost local at $k+L$ unitary, producing a $G$-invariant state $\tilde{\Theta}^{(k-L)}$, followed by a conjugation with a unitary strictly local on $[k,k+L)$. Let $|\tilde{\Xi}^{(k)}\ral$ and $|\tilde{\Theta}^{(k-L)}\ral$ be vectors representing the states $\tilde{\Xi}^{(k)}$ and $\tilde{\Theta}^{(k-L)}$ in the GNS Hilbert space $\CH_{\tilde{\Xi}^{(k)}}$. The vector $\tilde{\Xi}^{(k)}$ is the vacuum vector and thus is $G$-invariant. The vector $|\tilde{\Theta}^{(k-L)}\ral$, in general, transforms in a one-dimensional representation of $G$. We choose $W_{k}$ to be the dual of this representation, while $W_l$ for $k-L < l < k$ to be the trivial representation. 

For each $k\in\ZZ$ let $\Ups^{(k)}$ be a pure  factorized state on $\SA''$ which coincides with $\psi''$ on $(-\infty,k]$ and whose restriction to $\SA''_j,$ $j\in (k,+\infty)$, is given by $\langle t_j|\cdot |t_j\rangle$. This state is $G$-invariant and restricts to  pure factorized $G$-invariant states both on $(-\infty,k]$ and on $(k,+\infty)$. Thus $\Upsilon^{(k-L)}$ can be obtained from $\Upsilon^{(k)}$ by a conjugation with a strictly local at $k$ unitary. Therefore the vector $|\Upsilon^{(k-L)}\rangle$ in the GNS Hilbert space of the state $\Upsilon^{(k)}$ transforms in the representation $W_k$. Consequently, the vector $|\tilde{\Theta}^{(k-L)}\otimes \Upsilon^{(k-L)}\rangle\in\CH_{\tilde{\Xi}^{(k)}\otimes \Upsilon^{(k)}}$ is $G$-invariant, just like the vacuum vector $|\tilde{\Xi}^{(k)}\otimes \Upsilon^{(k)}\rangle$. Therefore their arbitrary linear combinations are also $G$-invariant. 

Let $\Omega_s$, $s \in [0,1]$, be the path of of vector states corresponding to the normalization of the path of vectors  $s|\tilde{\Xi}^{(k)}\otimes \Upsilon^{(k)}\rangle + (1-s)|\tilde{\Theta}^{(k-L)}\otimes \Upsilon^{(k-L)}\rangle$. By Corollary \ref{cor:almcstatesSRE} there is an LGA $\alpha_{P}$ generated by an almost local observable $P(s)$, such that $\Omega_s = \alpha_P(s)(\tilde{\Xi}^{(k)}\otimes \Upsilon^{(k)})$.
Let $P^G(s)$ be the observable obtained from $P(s)$ by averaging  over the group action. Then $P^G(s)$ generates an automorphism $\alpha_{P^G}$ which is $G$-equivariant, and since the state $\Omega_s$ is $G$-invariant for all $s$, we still have
$\Omega_s=\alpha_{P^G}(s)(\tilde{\Xi}^{(k)}\otimes \Upsilon^{(k)})$. The automorphism $\alpha_{P^G}$ is a
conjugation by some $G$-invariant almost local unitary $e^{i \CG_k}$ which is $h'$-localized at $k-L$ for some $h'(r)=\Or$ which depends only on $f(r).$ The state $\tilde{\Theta}^{(k-L)}\otimes \Upsilon^{(k-L)}$ can be obtained from $\tilde{\Xi}^{(k)}\otimes \Upsilon^{(k)}$ by conjugation with this observable.

A similar averaging argument shows that  $\tilde{\Xi}^{(k-L)}\otimes \Upsilon^{(k-L)}$ can be obtained from $\tilde{\Theta}^{(k-L)}\otimes \Upsilon^{(k-L)}$ by a conjugation with a strictly local on $(k-L,k]$ $G$-invariant observable. Therefore, in the same way as in the proof of Lemma \ref{lma:12invertibleSRE} and Theorem \ref{thm:invertibleSRE}, by taking $L$ large enough we can construct a $G$-invariant LGA that disentangles $\Psi$ tensored with several factorized states. 
\end{proof}

\begin{corollarythm}
    The index defines a group isomorphism between $\Phi^*_G$ and $H^2(G, U(1)).$
\end{corollarythm}
\begin{proof}
In section 3 we have shown that the index defines a group homomorphism from $\Phi^*_G$ to $H^2(G, U(1)).$ The MPS construction in section 3 implies that the homomorphism is surjective. Moreover, the theorem shows that this homomorphism has trivial kernel. Therefore, it defines a group isomorphism. 
\end{proof}

\begin{corollarythm}
    Every $G$-invariant invertible system is $G$-invertible.
\end{corollarythm}

\begin{proof}
Let $(\SA,\psi)$ be a $G$-invariant invertible system. Construct another $G$-invariant system $(\SA',\psi')$ whose index is the group inverse to that of $(\SA,\psi)$. This can be done using to the Entangled Pair State construction in section 3. The stacked system $(\SA\otimes \SA',\psi\otimes \psi')$ has trivial index according to the stacking rule shown in section 3. The theorem implies that this stacked system is in the trivial $G$-invariant phase. Hence, $(\SA,\psi)$ is $G$-invertible.
\end{proof}

\noindent
{\bf Acknowledgements:}
We would like to thank P. Etingof and B.  Simon for advice. We are also grateful to Y. Ogata for drawing our attention to an error in Lemma 4.1 in the original version of the paper. This research was supported in part by the U.S.\ Department of Energy, Office of Science, Office of High Energy Physics, under Award Number DE-SC0011632. A.K. was also supported by the Simons Investigator Award.
B.Y. acknowledges the Caltech mathematics department for a graduate fellowship awarded in fall 2020. 
\\

\noindent
{\bf Data availability statement:}
Data sharing is not applicable to this article as no new data were created or analyzed in this study.

\appendix
\numberwithin{equation}{section}

\section{Stable equivalence of $G$-invariant factorized pure states}

In this section we study $G$-invariant factorized pure states defined by (\ref{factorizedpure}) where the vectors $v_j$ may transform in non-trivial one-dimensional representations of $G$. We will show that all such states are in the trivial $G$-invariant phase.

Consider first a 1d system where for $j\neq 0$  $\SV_j=\CC$ is the trivial representation of $G$,  while $\SV_0=\SW$ is a finite-dimensional representation containing a unit vector $w$ transforming in a non-trivial one-dimensional representation of $G$.  Consider a $G$-invariant factorized pure state where $v_0=w$ (all other $v_j$ are unique up to a scalar multiple). Physically, this corresponds to a non-trivial $G$-invariant ground state of a 0d system regarded as a $G$-invariant state of a 1d system. We are going to show that this 1d state is in a trivial $G$-invariant stable phase. 

Without loss of generality we may assume that $\SW$ contains a $G$-invariant vector $w'$. Indeed, if this is not the case, we can tensor the above system with a similar system where $\SW$ is replaced with $\SW^*\oplus \CC\cdot e$, where $G$ acts on the second summand by the trivial representation, and $v_0=e$. The auxiliary system is in the trivial $G$-invariant phase, so this does not affect the $G$-invariant phase of the system we are interested in. Then the composite system has $\SV_0=\SW\otimes\SW^*\oplus \SW\otimes \CC\cdot e$, with a $G$-invariant factorized pure state corresponding to $v_0$ in the second summand. The first summand now contains a $G$-invariant vector $w\otimes w^*$. 

Consider now an auxiliary system $\SA'$ where  $\SV'_j=\CC$ is the trivial representation for $j\leq 0$ and $\SV'_j=\SW^*\otimes\SW$ for $j>0$. The algebra of local observables $\SA\otimes\SA'$ has the form
\beq
\End(\SW)\otimes \End(\SW^*\otimes\SW)\otimes \End(\SW^*\otimes\SW)\otimes\ldots
\eeq
We pick a $G$-invariant factorized pure state $\psi'$ on $\SA'$  defined by the condition that for any $A\in\SA'_j=\End(\SV'_j)$, $j>0$, it is a vector state corresponding to $w^*\otimes w$ which is $G$-invariant. Consider now the state $\psi\otimes\psi'$ on $\SA\otimes\SA'$. As $G$ acts trivially on the vector
state in $\SA_j'$ for all $j$, the composite state belongs to the same $G$-invariant phase as $(\psi,\SA)$. 

After re-writing the algebra of observables as 
\beq
\End(\SW\otimes\SW^*)\otimes \End(\SW\otimes\SW^*)\otimes \ldots
\eeq
it is easy to see that the state $\psi\otimes\psi'$ is related by a $G$-equivariant LGA to a special factorized pure state. Indeed, since $\SW$ contains a $G$-invariant vector $w'$,  $\SW\otimes\SW^*$ contains a $G$-invariant 2-plane spanned by vectors $w\otimes w^*$ and $w'\otimes {w'}^*$. Let $U$ be a unitary operator on $\SW\otimes\SW^*$ which acts by identity on the orthogonal complement of this plane and rotates $w\otimes w^*$ into $w'\otimes {w'}^*$. Consider a local unitary circuit on $\SA\otimes \SA'$ which acts on $\SA\otimes\SA'$ by conjugation with $U\otimes U\otimes U\otimes\ldots$. It maps $\psi\otimes\psi'$ to a factorized pure state on $\SA\otimes\SA'$ with $v_0=w'$ and $v_j={w'}^*\otimes w'$ for $j>0$. All these vectors are $G$-invariant.
 Thus $\psi$ is a special $G$-invariant factorized pure state.\footnote{This argument is a version of the Eilenberg swindle.}

In general, we can write the algebra $\SA$ as a tensor product of sub-algebras $\SA_{\geq 0}$ and $\SA_{< 0}$ corresponding to $j\geq 0$ and $j<0$. Since the state $\psi$ on $\SA$ is assumed factorized, it is sufficient to consider the restriction of $\psi$ to one of these sub-algebras, say $\SA_{\geq 0}$. Then we apply the argument of the above paragraph to each of the factors $\SA_j$ separately. Note that the auxiliary system $\SA'$ in this case has $\log d'_j$ growing with $j$ even if the dimension $d_j$ of $\SA_j$ is bounded. However, it is easy to see that if $\log d_j$ grows at most as a power of $j$, then so does $\log d'_j$. Thus it is still true that $\psi$ is in the trivial $G$-invariant phase.

\section{Local computability of the index}

An important property of the index is its local computability, i.e. that one can compute $\nu(g,h)$ up to $\OL$ accuracy while having access only to a disk of radius $L$. Let us fix a disk $\Gamma_{L}=[-L,L]$, and let $R^g_L$ be a unitary $\prod_{j \in \Gamma_L} R_j(g)$. By Theorem \ref{thm:invertibleSRE} for any invertible state $\psi$ there is a pure factorized state $\psi'_0$ on $\SA'$ (with a trivial $G$ action) such that $\psi\otimes\psi'_0$ can be produced from a pure factorized state $\Omega$ on $\SA\otimes\SA'$  by some LGA $\alpha_F$ for an $f$-local $F$ for some MDP function $f(r)=\Or$. In what follows we redefine $\psi$ to be the SRE state $\psi \otimes \psi'_0$.

First, note that by Lemma \ref{lma:aux2} and Corollary \ref{cor:almcstatesSRE} we have
\beq
U_{>j}^g | \psi \ral = \Pi(\CV_{>j}^g)  | \psi \ral .
\eeq
for some observable $\CV_{>j}^g$ which is $g$-localized at $j$ for some function $g$ that depends on $f$ only. We can find an observable $\CV^g_L$ local on $\Gamma_{L}$ such that $\| \CV^g_{>j} - \CV^g_L \| = \OL$.
The index can be computed as
\begin{multline}
\nu(g,h) = \lal \psi| (U^{gh}_{>j} )^{-1} U^g_{>j} U^h_{>j}| \psi \ral =\lal \psi|\Pi(\CV^{gh}_{>j} )^{-1} U^g_{>j}\Pi(\CV^h_{>j}) | \psi \ral \\
=\lal \psi| \Pi(\CV^{gh}_{>j} )^{-1} \Pi(\rho^g_{>j}(\CV^h_{>j})) \Pi(\CV^g_{>j})| \psi \ral = \psi \l (\CV^{gh}_{>j} )^{-1} \rho^g_{>j}(\CV^h_{>j}) \CV^g_{>j} \r
\end{multline}
Therefore
\beq
\nu(g,h) = \psi \l (\CV^{gh}_{L} )^{-1} R^g_L \CV^h_{L} (R^g_L)^{-1} \CV^g_{L} \r + \OL.
\eeq
Note that the $\OL$ term depends on $f(r)$ only, and by taking $L$ large enough we can compute the index with any given accuracy. Therefore if we have an interpolating $G$-invariant invertible state $\psi$ which is $f$-close to a $G$-invariant invertible state $\psi_1$ on the left half-chain and $f$-close to a $G$-invariant invertible state $\psi_2$ on the right half-chain, then the indices of $\psi_1$ and $\psi_2$ must be the same. In particular, a non-trivial index for $\psi_1$ is an obstruction for the existence of such an interpolation between $\psi_1$ and a pure factorized state $\psi_2$.

\section{Multiplicative Lieb-Robinson bound}

\begin{lemma}
Let $\CA_i, i=0,1,2,\ldots$ be quasi-local observables such that $\sum_i \|\CA_i - 1\|$ converges. Then the sequence of observables $C_n=\prod_{i=1}^n \CA_i$ is norm-convergent. 
\end{lemma}
\begin{proof}
Let $\CalC_n = \prod_{i = 0}^n \CA_i$ and $\CB_i = \CA_{i+1} - 1$. Then \beq \label{difference}
\CalC_n = \sum_{i = 0}^ {n-1} \CalC_i \CB_i + \CalC_0.
\eeq 
This implies 
$\|\CalC_n\|\leq \sum_{i = 0}^{n-1}\|\CB_i\| \|\CalC_i\| + \|\CalC_0\|$, which by the discrete Gronwall inequality \cite{Gronwall} implies
$$\|\CalC_n\| \leq \|\CalC_0\| \exp\left(\sum_{i = 0}^{n-1}\|\CB_i\|\right)\leq \|\CalC_0\| \exp\left(\sum_{i = 0}^{\infty}\|\CB_i\|\right) < \infty.$$
Thus the right hand side of equation (\ref{difference}) converges in norm as $n\ra\infty$.
\end{proof}
\begin{corollary}\label{cor:infiniteproduct}
Let $\CV_k$, $k=0,1,2,\ldots$ be local unitaries localized on $[(-k-1/2)L,(k+1/2)L]$  and  satisfying $\|\CV_k-1\| \leq h((k+1/2)L)$ for some MDP function $h(r)=\Or$. Then product $\CV_0 \CV_1 \CV_2...$ exists. Furthermore, it is $f$-local at $0$ for some MDP function $f(r)= \Or$ determined by $h$.
\end{corollary}
\begin{proof}
As $h(r)=\Or$ and $\|\CV_k-1\| \leq h((k+1/2)L)$, $\sum_k\|\CV_k-1\| \leq \sum_k h((k+1/2)L)$ converges. Therefore, the product $\CalC_\infty=\CV_0 \CV_1 \CV_2\ldots$ exists by the lemma above.
Let $\CalC_n = \CV_0 \CV_1 \CV_2...\CV_n$ and $\CB_i = \CV_{i+1} - I$. Then
\beq
\CalC_n = \sum_{i = 0}^{n-1} \CalC_i \CB_i + \CalC_0.
\eeq
For any $\CA \in \SA_j$, 
\begin{multline}
    \|[\CalC_\infty, \CA]\| = \|[\sum_{i = 0}^{\infty} \CalC_i \CB_i + \CalC_0, \CA]\| \leq  \sum_{(i+\frac12)L > j}\|[ \CalC_i \CB_i, \CA]\|\\ \leq 2\| \CA\|\sum_{(i+\frac12)L > j}\| \CB_i\| \leq 2\| \CA\|\sum_{(i+\frac12)L > j}h((i+1/2)L).
\end{multline}
Thus we may let $f(r)= \sum_{s > r} h(s) = \Or$.
\end{proof}

In what follows for any sequence of automorphisms $\alpha_n$, $n\in\ZZ,$ we let $\overrightarrow{\prod_n}\alpha_n$ be the formal expression $\ldots \circ\alpha_{-1}\circ\alpha_0\circ\alpha_1\circ\ldots$. Similarly, we denote by $\overleftarrow{\prod_n}\alpha_n$ the formal expression $\ldots \circ\alpha_{1}\circ\alpha_0\circ\alpha_{-1}\circ\ldots$. These expressions are well-defined automorphisms if all but a finite number of $\alpha_n$ are identities. The following lemma describes a class of situations when the formal expressions make sense even if an infinite number of $\alpha_n$ are nontrivial.

\begin{lemma} \label{lma:1dLGA}
Let $\Lambda=\ZZ\subset\RR$. For any MDP function $f(r)=\Or$ there is $L\in\NN$ such that any ordered composition $\overrightarrow{\prod_{n=-\infty}^{\infty}} \alpha_{\CB^{(n)}}$ of LGAs generated by  $f$-local at $j=n L$ observables $\CB^{(n)}$ for $n \in \ZZ$ is an LGA.
\end{lemma}

\begin{proof}
First, note that it is enough to show this for $\overrightarrow{\prod_{n=0}^{\infty}} \alpha_{\CB^{(n)}}$. Second, to prove the latter it is enough to show that with an  appropriate choice of $L$ for any $f$-local at 0 observable $\CA$ the observable $(\overleftarrow{\prod_{n=1}^{N}} \alpha_{\CB^{(n)}})(\CA)$ is almost local at 0 with localization depending on $f$ only (in particular, independent of $N$), and that as $N \to \infty$ it converges in norm to some element of $\SA$.

Let $\CU^{(n)} := e^{i \CB^{(n)}}$ be a unitary that corresponds to $\alpha_{\CB^{(n)}}$. It can be represented as a product $(\CV^{(n)}_0 \CV^{(n)}_1 \CV^{(n)}_2...)$ of strictly local unitaries $\CV^{(n)}_k$ on $B_n((k+\frac12)L):=[(n-k-\frac12)L,(n+k+\frac12)L]$, so that $\|\CV^{(n)}_k-1\| \leq h((k+\frac12)L)$ for some MDP function $h(r)=\Or$ that depends on $f(r)$ only. This is achieved by letting $(\CV^{(n)}_0 \CV^{(n)}_1 ... \CV^{(n)}_k) = e^{i\CB|_{B_n((k+\frac12)L)}}$.

Since conjugation of a strictly local observable $\CA$ with a unitary $\CU$ strictly local in the localization set  of $\CA$ does not change the property $\|\CA-1\|<\eps$ and preserves the localization set, we can rearrange unitaries $\CV^{(n)}_k$ in the product 
\beq \label{eq:firstpr}
(\CV^{(1)}_0 \CV^{(1)}_1 \CV^{(1)}_2...) (\CV^{(2)}_0 \CV^{(2)}_1 \CV^{(2)}_2...) ... (\CV^{(N)}_0 \CV^{(N)}_1 \CV^{(N)}_2...)
\eeq
in the following order:
\beq \label{eq:secondpr}
(\tilde{\CV}^{(1)}_0) (\tilde{\CV}^{(2)}_0 \tilde{\CV}^{(1)}_1) (\tilde{\CV}^{(3)}_0 \tilde{\CV}^{(2)}_1 \tilde{\CV}^{(1)}_2) ... (\tilde{\CV}^{(n)}_0 \tilde{\CV}^{(n-1)}_1 ... \tilde{\CV}^{(1)}_{n-1}) ...,
\eeq 
where $\tilde{\CV}^{(n)}_k$ is obtained from   $\CV^{(n)}_k$ by conjugation with $\CV^{(m)}_l$ with $m,l$ satisfying $n+1\leq m \leq n+k$ and $0\leq l\leq n+k-m$. Importantly,  $\tilde{\CV}^{(n)}_k$ is strictly local on the same interval as $\CV^{(n)}_k$ and still satisfies $\|\tilde{\CV}^{(n)}_k - 1\| \leq h((k+\frac12)L)$.  The infinite product  eq. (\ref{eq:secondpr}) is a well-defined almost local observable by Corollary \ref{cor:infiniteproduct}. Indeed, $\|\tilde{\CV}^{(n)}_0 \tilde{\CV}^{(n-1)}_1 ... \tilde{\CV}^{(1)}_{n-1}-1\|\leq \sum_{i = 1}^{n}\|\tilde{\CV}^{(i)}_{n-i}-1\|$ by repeatedly applying the inequality $\|\CA\CB-1-\CA+\CA\|\leq \|\CA\|\|\CB-1\|+\|\CA-1\|.$ Since for any fixed $N$ we have $\CV^{(n)}_k = 1$ for $n>N$,  $\|\tilde{\CV}^{(n)}_0 \tilde{\CV}^{(n-1)}_1 ... \tilde{\CV}^{(1)}_{n-1}-1\|\leq \sum_{i = 1}^{N}\|\tilde{\CV}^{(i)}_{n-i}-1\|\leq \sum_{i = 1}^{N}h((n-i+\frac12)L)$ satisfies the assumption of Corollary \ref{cor:infiniteproduct}. After this rearrangement the infinite product eq. (\ref{eq:secondpr}) still converges to the same unitary observable as eq. (\ref{eq:firstpr}).

Let $\tilde{\CU}^{(n)} = \tilde{\CV}^{(n)}_0...\tilde{\CV}^{(1)}_{n-1}$. We can represent $\CA = \sum_{p=0}^{\infty} \CA_p$ with $\sum_{p=0}^{n} \CA_p = \CA|_{B_{0}(n+1/2)}$. Let $\CA^{(0)}_p:=\CA_p$, $\CA^{(n)}_p:=\sum_{k=0}^{n-1} \tilde{\CU}^{(n)*} [\CA^{(k)}_p, \tilde{\CU}^{(n)}]$. Note that $\CA^{(n)}_p \in \SA|_{B_0(p+\frac12)}$ for $n \leq p$, and $\CA^{(n)}_p \in \SA|_{B_0(n+\frac12)}$ for $n>p$. Therefore we have
\begin{multline} \label{eq:Aestimate}
\|\CA^{(n)}_0\| =  \sum_{k=0}^{n-1} \|[\CA^{(k)}_0,\tilde{\CU}^{(n)}]\| \leq \sum_{k=0}^{n-1} \sum_{l>\frac{n-k-1}{2}}^{n-1} \|[\CA^{(k)}_0,\tilde{\CV}^{(n-l)}_l]\| \leq \\ \leq \sum_{k=0}^{n-1} \sum_{l > \frac{n-k-1}{2}}^{\infty}  2 \|\CA^{(k)}_0\| h((l+\frac12)L) \leq 2 \sum_{k=0}^{n-1} \|\CA^{(k)}_0\| g_{n-k}
\end{multline}
where $g_n := g(n L/2)$ for $g(n) := \sum_{l \geq n}^{\infty} h(l)$.

Any MDP function $g(r)=\Or$ can be upper-bounded by a reproducing MDP function $\tilde{g}=\Or$  for lattice $\Lambda \subset \RR$ \cite{hastings2010quasi}, i.e. an $\Or$ MDP function satisfying
\beq
\sup_{j,k \in \Lambda} \sum_{l \in \Lambda} \frac{\tilde{g}(|j-l|)\tilde{g}(|l-k|)}{\tilde{g}(|j-k|)} < \infty.
\eeq
We can further upper-bound $\tilde{g}(r)$ by a reproducing $\Or$ MDP function $g'(r)=A \tilde{g}(r)^{\alpha}/r^{\nu} = \Or$ for some constants $A$, $0<\alpha<1$ and $\nu > d$. Since $1/r^{\nu}$ is also reproducing, we have

\beq
A^2 \sum_{k=1}^{n-1} \frac{\tilde{g}(k L/2)^{\alpha} \tilde{g}((n-k) L/2)^{\alpha}}{(kL/2)^{\nu} ((n-k)L/2)^{\nu}} < \frac{C}{L^{\nu}} A \frac{\tilde{g}(n L/2)^{\alpha}}{(nL/2)^{\nu}},
\eeq
and therefore for $g'_n := g'(nL/2) \geq g_n$ we have $\sum_{k=1}^{n-1} g'_k g'_{n-k} < (C/L^{\nu}) g'_n$
for some constant $C$. For $L$ sufficiently large, we have $(C/L^{\nu})<1/2$. Therefore, for such $L$ after setting $a_n = 2 n g'_n$ we get
\begin{multline} 
2(g_n \cdot 1 + g_{n-1} a_{1}+ g_{n-2} a_{2} + ... + g_{1} a_{n-1}) \leq \\ \leq 2(g'_n \cdot 1 + 2 g'_{n-1} g'_{1} + 4 g'_{n-2} g'_{2} + ... + 2(n-1) g'_{1} g'_{n-1}) \leq 2 n g'_n = a_n.
\end{multline}
Together with eq. (\ref{eq:Aestimate}) this implies that $\|\CA^{(n)}_0\|/\|\CA_0\|$ can be upper-bounded by $a_n = \On$, and the sequence $\sum_{k=0}^{n}\CA^{(k)}_0$ converges in norm to some almost local observable. By construction, it is $h$-localized at $0$ for some MDP function $h(r)=\Or$ which depends on $f$ only.

In the same way one can estimate the norms of $\CA^{(n+p)}_p$ for $n,p>0$ and bound $\|\CA^{(n+p)}_p\|/\|\CA_p\|$ by a sequence $a_n = \On$. Together with $\|\CA_p\| = \| \sum_{q=0}^{p}\CA^{(q)}_p\|$ that ensures convergence of $\sum_{p=0}^{\infty} \sum_{n=0}^{\infty} \CA^{(n)}_p$ to some almost local at 0 observable whose localization depends on $f$ only.

\end{proof}

\printbibliography

\end{document}